\documentclass[]{article}
\usepackage{authblk}
\usepackage[a4paper, total={6in, 8in}]{geometry}

\usepackage{epigraph}
\usepackage{bbm}
\usepackage{lettrine}
\usepackage{tikz-cd}
\usepackage{amsthm,amsmath,amssymb}
\usepackage{amsmath}
\usepackage{graphicx}
\usepackage{caption}
\usepackage{changes}
\usepackage{soul}
\usepackage{tikz}
\usepackage{float}
\usepackage{tikz-cd}
\usetikzlibrary{arrows}
\usepackage{scalefnt}
\usepackage{makecell}
\usepackage[nomessages]{fp}
\usepackage{enumerate}
\usepackage{tikz-qtree}
\usetikzlibrary{shapes,positioning,chains,fit,calc}
\usepackage{rotating}
\usepackage{cancel}
\usepackage{algorithm}
\usepackage[noend]{algpseudocode}
\usepackage{multirow}
\usepackage[section]{placeins}
\usepackage{xspace}
\newcommand{\soc}{\textsf{SOC}\xspace}
\usepackage{lineno}
\usepackage{hyperref}

\newcommand{\NN}{\mathbb{N}}

\newtheorem{theorem}{Theorem}[section]
\newtheorem{lemma}[theorem]{Lemma}

\newtheorem{definition}{Definition}

\title{Centralities for Networks with Consumable Resources}
\author[1]{Hayato Ushijima-Mwesigwa \thanks{corresponding author: \texttt{hushiji@g.clemson.edu}}}
\author[2]{Zadid Khan}
\author[2]{Mashrur A. Chowdhury}
\author[1]{Ilya Safro}

\affil[1]{School of Computing, Clemson University, Clemson, South Carolina, USA}
\affil[2]{Department of Civil Engineering, Clemson University, Clemson SC, USA}

\date{}                     

\begin{document}

\maketitle



\begin{abstract}
	
	Identification of influential nodes is an important step in understanding and controlling the dynamics of information, traffic and spreading processes in networks. As a result, a number of centrality measures have been proposed and used across different application domains. At the heart of many of these measures, lies an assumption describing the manner in which traffic (of information, social actors, particles, etc.) flows through the network. For example, some measures only count shortest paths while others consider random walks.
	This paper considers a spreading process in which a resource necessary for transit is partially consumed along the way while being refilled at special nodes on the network. Examples include fuel consumption of vehicles together with refueling stations, information loss during dissemination with error correcting nodes, and consumption of ammunition of military troops while moving.
	We propose generalizations of the well-known measures of betweenness, random walk betweenness, and Katz centralities to take such a spreading process with consumable resources into account. In order to validate the results, experiments on real-world networks are carried out by developing  simulations based on well-known models such as Susceptible-Infected-Recovered and congestion with respect to particle hopping from vehicular flow theory. The simulation-based models are shown to be highly correlated to the proposed centrality measures.\\
	\noindent{\bf Reproducibility: }Our code, and experiments are available at {\url https://github.com/hmwesigwa/soc\_centrality}
	\vspace{0.5cm}
	
\end{abstract}

\section{Introduction}
Spreading processes are ubiquitous throughout science, nature, and society \cite{strogatz2001exploring,tao2006epidemic,pastor2001epidemic}. 
These include, spreading of infectious diseases \cite{keeling2011modeling}, computer viruses \cite{kephart1997fighting}, cascading failures \cite{motter2004cascade}, traffic congestion \cite{li2015percolation}, opinion spreading \cite{liu2007opinion,bettencourt2006power}, and reaction-diffusion processes \cite{colizza2007reaction}. 
Understanding a nodes' spreading influence is fundamental for a wide variety of applications such as epidemiology \cite{diekmann2000mathematical,keeling2011modeling}, viral marketing \cite{watts2007viral,leskovec2007dynamics}, collective dynamics \cite{albert2002statistical,boccaletti2006complex,barrat2008dynamical} and robustness of networks \cite{albert2000error,newman2003structure,cohen2001breakdown} and so forth.
Whereas many centrality measures were originally developed for social networks, some of them have subsequently been adapted to quantify the importance of nodes in epidemiological spreading processes \cite{kitsak2010identification,ghosh2012rethinking,vsikic2013epidemic, bae2014identifying}. This is partly due to the fact that most centralilty measures have simple assumptions, thus these measures are often intuitive and interpretable for a given application. Moreover, most popular centrality measures are based on variants of paths and eigenvector computations which explain paradigms in spreading models.
Popular centrality measures include degree, closeness, betweenness, current-flow, PageRank, eigenvector and Katz centralities \cite{freeman1978centrality,newman2005measure,brandes2005centrality,page1999pagerank,bonacich1987power,bonacich1991simultaneous,katz1953new}. All these measures make an implicit assumption about the process in which a commodity (e.g., information, vehicles, or infection) flows in the network. Typically, closeness and betweenness assume flow on geodesic paths, while PageRank, eigenvector and Katz centrality model flow via random walks. The extent to which a centrality measure can be interpreted for a given application depends on whether or not the assumed flow characteristics are a good representation of what is actually flowing in the network.

In this work, we consider a flow process in which a resource essential for flow is consumed along the way and can be refilled at specially assigned nodes in order to ensure that a flow process is not terminated. For example, a vehicle consumes fuel as it travels in a network that has refueling station nodes. In another domain, information requires updating or refreshing while it moves over the network. For example, in real information and social networks, rumors and gossips often die out if not refreshed \cite{moreno2004dynamics} and forgetting rates are considered in models \cite{zhao2013rumor}. Not much of the existing work models well-known concepts on networks by taking into account consumable resources.

For simplicity, we model the resource consumption as a discrete process that limits the number of steps the flow process can take without refilling the resource. For a graph underlying network of interest, $G=(V,E)$, the parameter $\kappa$ represents the number of steps a process can take without a consumed resource being refilled, and $\Omega \subset V$ represents the refilling nodes. One of the most intuitive and important modern applications of this process is the in-motion recharging of electric vehicle batteries that is anticipated to be broadly implemented in future. We borrow terminology from the charge level of batteries for electric vehicles and refer to the currently available resource as the \emph{state of charge} (\soc). Thus, $\kappa$ represents the full \soc value. In the next section, a list of related real-word applications is given.

\subsubsection*{Our contribution}
In this work, we study a process where a commodity (such as information, and traffic) is flowing (or spreading) in a network while consuming a resource necessary for flow, and being refilled at special nodes. We give a list of potential applications that have a similar flow (or spreading) process. In order to estimate a nodes' spreading influence, we generalize the measures of Katz, betweenness, and random-walk betweenness centralities (including its generalization for directed graphs) and show how they can be computed. Lastly, we present different models to simulate the spreading processes and show that the generalized centrality measures are highly correlated to the simulation-based models.

\section{Applications}
\emph{Transportation networks:} A natural and motivating application of such a process is in transportation networks, in particular, road networks. Let a node in a road network represent a road segment. An edge between two road segments exists if they are physically adjacent to each other (as prolongation of each other or with a real intersection). 
In this case, $\Omega \subset V$ represents the nodes with refueling stations, and $\kappa$ - the maximum distance a vehicle can travel without refueling. 
In particular, we can also consider electric vehicle road networks equipped with wireless charging lanes, where a whole lane can be turned into a charging infrastructure. 
This technology has seen tremendous growth over the last couple of years with test sites already in place \cite{jang2012optimal}. However, setting up this technology will come with a heavy price tag for a city with a limited budget, thus tools must be developed to analyze these networks beforehand. While several research studies have carried out for identifying the optimal locations for the deployment of wireless charging lanes \cite{riemann2015optimal,chen2016optimal,ushijima2017optimal,khan2017utility,khan2019wireless}, these studies often make different assumptions while solving different optimization objectives and constraints. Having an independent tool to analyze the network is thus also necessary in order for deployment strategies to be compared.

\emph{Peer-to-Peer Networks:}
Peer-to-Peer (P2P) information exchange systems  \cite{basu2013state} have gained popularity over the past two decades. One of the challenges in P2P systems is searching for content on the network. Gnutella is a popular open, and decentralized file-sharing protocol in P2P networks \cite{wang2007analyzing}. The Gnutella protocol works as follows \cite{ripeanu2002mapping}: 
\begin{enumerate}
	\item A node (computer) $v$ connects to the Gnutella network by connecting to a set of one or more nodes, $U$, already in the network. Then $v$ announces its existence to all nodes in $U$. 
	\item The nodes in $U$ announce to all their neighbors that $v$ has joined the network, which also announce to their neighbors, and so forth.
	\item Once all nodes are aware of $v$'s existence, it can make a query on the network.
\end{enumerate}
Popular methods of message propagation for a given query issued by a node include such methods as flood-based, and random walks routing algorithms \cite{tsoumakos2006analysis}. A global time-to-live (TTL) parameter represents the maximum number of steps (also known as hops) a query can take before it gets discarded. In a flood-based routing algorithm, a querying node contacts all its neighbors, who then contact all their neighbors, and so forth. The process stops after each message has taken TTL number of steps. This simplistic method produces a huge overhead by contacting many nodes. In the random walks routing algorithm, the querying node randomly chooses a subset of its neighbors and sends each of them $k$ messages, for some $k$. Each of these messages starts its own random walk in the network that is terminated after TTL steps. Other termination conditions exist, however, they are not relevant for this work.  The random walks routing algorithm greatly reduces the message passing throughout the network, with other advantages such as local load balancing, since no nodes are favored over others during message propagation. However, depending on the network topology, success rates could vary significantly. These two routing algorithms are sometimes referred to as blind search methods. On the other hand, informed search methods include methods that for example, take advantage of previous queries making better decisions for message passing, and in an ideal scenario, a message could then take the shortest path to a target node. 

In the Gnutella network, once a node receives a message, it first reduces the TTL counter of the message before forwarding it. The TTL parameter is intuitively equivalent to the \soc parameter $\kappa$ in this paper while the set of nodes in $\Omega$ represent nodes that reset the TTL counter before forwarding the message. These, for example, could be compromised nodes. In the analysis of the Gnutella network, an interesting question is determining the most important nodes in the network, which could be the nodes that receive the most traffic.

\emph{Social Networks (online):} The popularity and complexity of online social networks (OSN) have seen a tremendous growth in the last two decades and will continue to grow. In OSNs such as Twitter or Facebook, a user shares information which can be viewed by other users he/she is connected to. Centrality measures are often used to identify influential users within an OSN. However, in most centrality measures, all the users in the network are assumed to exhibit similar behavioral features. In other words, they all have the same desire and motivation to share knowledge in the network. In most cases, all the users are assumed to be \emph{active} users, users that are willing and motivated to share their knowledge. Many studies have been carried and show that this is not the case. In fact, most users, while being beneficiaries of the content being shared, actually do not share information themselves. The two different types of users are often referred to as \emph{posters} and \emph{lurkers} \cite{marett2009decision,lai2014knowledge,preece2004top,schlosser2005posting}. Posters are defined in \cite{schlosser2005posting} as the active users who share their experiences and create content on the internet while lurkers are defined as the passive users who do not necessarily create any content. 
The analysis of lurkers in social networks is now an active area of research \cite{interdonato2016trust}. In order to have a better understanding of online social networks, it is important to understand how information and content is shared over the network. It has been reported that lurkers are the majority in many online communities. The percentages of lurkers in an online community varies across different studies with some giving estimates as high as 90 percent of the total users \cite{van20141}. Sometimes referred to as the \emph{participation inequality principle} \cite{interdonato2016trust}, only a small percentage of users actually contribute to the online content while the rest never do. Lurkers are not registered users who do not use their account, they can share information in subtle ways. For example, Facebook has the "like" feature and thus user's contacts can see the information he/she liked. In this work, we take the posters to be members of the set $\Omega$ and assume that a piece of information is coupled with a \emph{momentum} or \emph{penetrating power}. Following the analogy of a battery charge in  EV, the momentum of a piece of information has the power to drive the information for a limited number of steps, $\kappa$,  before it dies out. However, if it reaches a node in $\Omega$, it regains its momentum.  




\emph{Personalized web ranking:} Consider a person surfing the web at random, however, with a topic of interest in mind. The surfer begins at a web page $\omega_1$ and performs a random walk on the web graph. At each time step, the person proceeds from the current page $u$ to a randomly chosen web page that $u$ is linked to. If after $\kappa$ steps, the surfer has not found a web page of interest, the surfer starts the process again from $\omega_1$. However, if the surfer discovers a web page  $\omega_2$ that is relevant to the topic of interest, $\omega_2$ becomes the new restarting point and the process continues. If $\Omega$ is the set of web pages known beforehand, then centrality measure defined by such a process can be used to rank the web pages based on the ones in $\Omega$, giving a personalized page ranking strategy. In this class of applications, we can also mention random walk based similarity measures on graphs and hypergraphs \cite{shaydulin2017relaxation,fouss2007random,chen2011algebraic} that would benefit from introducing resource consumption restrictions for the distance of a random walk.




\section{Related Work}
Centrality as a way of analyzing social networks dates back to Bavelas \cite{bavelas1948mathematical}.
Since then, various methods of centrality have been proposed to quantify the importance of individuals in social networks. These measures have also been effectively used as tools to study networks in other diverse fields such as physics, biology, and engineering. 

Since our initial motivation for this research was related to the analysis of road networks, in this section, we first highlight how previous studies have used centrality measures to analyze road networks. Next, we briefly introduce and summarize studies on electric vehicle road networks. Lastly, based on the potential applications for the proposed centrality measures, we summarize different existing approaches for  possible applications.  

We briefly summarize an incomplete list of studies in which centralities have been used to study and analyze road networks. 
A road network pattern can be viewed as the geographical layout and structure of a network. A road network can be laid out in different patterns (for example, see the book \cite{southworth2013streets} for more information on road network patterns) 
which can affect traffic performance, travel behavior, and traffic safety. In Zhang et al., \cite{zhang2011centrality}, the betweenness centrality is computed to analyze and classify road network patterns. In particular, it is used to define a measure that can quantitatively distinguish between different pattern types. In Wang et al., \cite{wang2012understanding}, the authors use centrality measures to analyze road networks in urban areas and apply their findings to mitigate congestion. More specifically, they use large-scale mobile phone data, with detailed Geographic Information System data to detect types of road usage and determine the origins of the drivers. This information is used to build a bipartite network with nodes representing road segments and  \emph{driver sources}, which is then called the \emph{network of road usage}. Here, a driver source is a zone where the mobile phone user lives. This can be located using the mobile phone data.  Given a list of all driver sources, for each road segment $r$, the authors calculated the fraction of traffic flow on $r$ that was generated by each driver source. They then ranked the driver sources by their contribution to the traffic flow. Based on this information, an edge in the network of road usage exists between a road segment $r$ and the top-ranked source nodes that produce 80\% of $r$'s traffic flow. Finally, the betweenness centrality of a road segment $r$ in the road network, along with the degree centrality of $r$ in the network of road usage were used to classify and group the road segments in the network. Experiments carried out in the San Francisco Bay area and Boston area provides evidence to show that the findings could enable cities to tailor targeted strategies to reduce the average daily commute time. 


In Scheurer et al., \cite{scheurer2008spatial}, the authors used a wide variety of centrality measures such as betweenness, closeness and degree centrality to identify the positive and negative points of the public transportation networks from different perspectives such as coverage, connectivity and service levels. 
The collective human spatial movement behavior is explored in \cite{jiang2011agent}. The authors use, among others, PageRank and betweenness centrality  coupled with agent-based simulations to study the movement of pedestrians in London street network. 
In \cite{altshuler2011augmented}, the authors propose an estimation method for mobility prediction in transportation networks based on the betweenness centrality carrying out experiments on the Israeli transportation network. Other studies on transportation networks, where centrality plays a crucial role in the analysis include \cite{jayasinghe2015explaining,jiang2004structural,porta2006network,jayaweeracentrality,crucitti2006centrality,park2010social}. 

As cities move towards reducing their carbon footprint, EVs offer the potential to reduce both petroleum imports and greenhouse gas emissions. However, batteries in these vehicles have a limited travel distance per charge. This results in a major obstacle for EV widespread adaptation, namely, 
\emph{range anxiety}, the persistent worry about not having enough battery power to complete a trip. The emergence of EV wireless charging technology where a whole lane can be turned into a charging infrastructure provides itself as a potential solution to range anxiety. For a more detailed study of the design, application and future prospects of this technology, the reader is encouraged to see, for example, \cite{qiu2013overview,bi2016review,li2015wireless,lukic2013cutting, cirimele2014wireless,fuller2016wireless,vilathgamuwa2015wireless,ning2013compact,yan2006efficient,guimera2002optimal}.
With a heavy price tag, a deployment of this technology without a careful study can lead to inefficient use of limited resources. One of the main purposes of this paper is to provide a tool to study and analyze road networks with a given deployment of wireless charging lanes. In these EV road networks, we assume that in order for an EV to travel between any two nodes, it is possible that a vehicle may need to detour to get charged to arrive at its destination. We envision that our modified version of betweenness centrality can be used in studying these EV road networks in similar ways as the studies in the preceding paragraph. 


Studies analyzing social networks whose users can be categorized as posters and lurkers have recently been gaining attention. In \cite{tagarelli2013s,tagarelli2014lurking}, the authors propose centrality measures for ranking lurkers in social networks. In these works, no prior knowledge of whether a user is a lurker/poster or not is assumed. The authors define a topology-driven lurking framework to model the relationships from information-producer to information-consumer. As a result, lurkers are ranked based on only the topology of the network. In our applications of the proposed centrality measures, we assume prior knowledge of whether or not a user is a poster or lurker. The main basis for this assumption is that the network topology may not be a related to a user's desire to share information. For example, two users on Facebook may have the same number of connections, however, have very different desires to share or post information.

\section{Graph Model}
Let $G = (V,E)$ be an unweighted (directed or undirected) graph underlying a network of interest. The underlying assumption is that the commodity (such as information and moving vehicle) is flowing (or spreading) on $G$ while consuming a resource necessary for flow. The flow is limited to the nodes within a geodesic distance of at most $\kappa$ edges, 
$\kappa\in \NN$. In addition, there exists a subset of nodes, $\Omega \subset V$ that refill the resource. In other words, if the commodity passes through a node $u \in \Omega$, it can then spread further to nodes that are at most $\kappa$ edges, 
from $u$. This process is modeled by a directed graph with adjacency matrix $\mathcal{B}_{\kappa}$ (see below). 

Let $A$ be the adjacency matrix of $G$, and $|V| =n$. 
Define the state space of a commodity traversing $G$ as the set
\begin{equation*}
\mathcal{V} :=\{(u, i)| u \in V, 0 \leq i \leq \kappa \},
\end{equation*}
in which the state $(u,i)$ represents the event that the commodity is at node $u \in V$ with the current level of \soc at $i$. The transition from one state to another is modeled by a directed graph $\mathcal{G} = (\mathcal{V}, \mathcal{E})$, where $\mathcal{E} = \mathcal{E}_1 \cup \mathcal{E}_2$ with 

\begin{equation*}
\mathcal{E}_1 := \big\{ \big( (u,i), (v,\kappa )\big) \ | \ (u,v) \in E,  v \in \Omega  \big\},
\end{equation*}
and
\begin{equation*}
\mathcal{E}_2 := \big\{ \big( (u,i), (v,j)\big) \ | \ (u,v) \in E, \   i = j + 1, v \notin \Omega  \big\}.
\end{equation*}
The set $\mathcal{E}_1$ represents a transition where the current \soc is increased to $\kappa$ (refilled), while $\mathcal{E}_2$ represents a transition where the current \soc is reduced by 1 (consumed). The adjacency matrix of $\mathcal{G}$, denoted as $\mathcal{B}_{\kappa}$, is defined as follows. Let $J_{\Omega}$ be a diagonal $n\times n$ matrix given by

\begin{equation*}
[J_{\Omega}]_{ii} = \left\{\begin{array}{lr}
1, & \mbox{if } i \in \Omega \\
0, & \mbox{otherwise.} 
\end{array}\right\}
\end{equation*}
For ease of exposition, where it is clear, we drop the subscript in $J_{\Omega}$ and simply write $J$.
If $I$ the $n\times n$ identity matrix define the $n(\kappa  +1) \times n(\kappa+1)$  block matrix $\mathcal{B}_{\kappa}$ as 
\begin{equation}
\mathcal{B}_{\kappa} = \left[\begin{array}{llllllll}
AJ & A(I-J) & \textbf{0}& \dots& \dots& \dots& \textbf{0}\\\
AJ & \textbf{0} & A(I-J) &\textbf{0}& \dots& \dots &\textbf{0} \\
AJ & \textbf{0} &\textbf{0} & \ddots &\textbf{0}& \dots &\textbf{0} \\
\vdots & & & & \ddots & & \vdots\\
\vdots & & & &  & \ddots & \vdots\\
AJ& \textbf{0} &\textbf{0} & \dots &\textbf{0}& \dots A(I-J)\\
AJ& \textbf{0} &\textbf{0} & \dots &\textbf{0}& \dots &\textbf{0}\\
\end{array}\right]
\label{b_matrix}
\end{equation}
Then the block matrix $\mathcal{B}_{\kappa}$ defines a directed state space  graph $\mathcal{G} = (\mathcal{V}, \mathcal{E})$ that models the underlying flow process.
In order for a commodity to flow from node $s$ to $t$, it may be necessary to traverse one or more nodes in $\Omega$. With this in mind, we define a feasible walk to represent the walks in $G$ that the commodity can fully traverse.
\begin{definition}\label{def:fw}
	A walk $w$ in $G$ is a \emph{feasible walk} if a commodity starting with full \soc can traverse $w$.
\end{definition}
The proposed centrality measures in the following sections are based on computing the feasible walks in the network. 

\section{Katz Centrality}
The centrality measure proposed by Katz \cite{katz1953new} was originally intended to rank a node (i.e., an actor in a social system) influence within a social network according to the number of its contacts considering different path lengths to other nodes. 
Thus, the model takes into account not only the immediate neighbors of a node but also its neighbors of second-order, third-order and so on.  The computation of Katz centrality is based on random walks emanating from a node. In this section, we propose \soc-Katz centrality, that only takes feasible walks into account. 


\subsection*{Counting Number of Feasible Walks}
For an adjacency matrix $A$ of a graph, the $ij$th entry of the matrix $A^k$, $k\in \NN$, counts paths from $i$ to $j$ of length $k$. \emph{However, in our resource consumption model, not all of these walks are in fact feasible}. This leads to an interesting question of finding a matrix that represents the number of $i$-$j$ feasible walks.

Consider the matrix $\mathcal{B}_{\kappa}^k$. Assume that the index of nodes in $V$ and matrices $I,J,A, \mathcal{B}_{\kappa}$ start at 0. For $i,j \in V$, with $0 \leq i < n$, $0 \leq i' < n (\kappa + 1)$,  and $j \equiv  i' \pmod{n}$, the $ij$th entry of $\mathcal{B}_{\kappa}^k$ gives the number of walks from node $i$ to $j$  completing with a different SOC. The destination SOC is given by the value $\lfloor i'/n \rfloor$. This implies that the number of walks from $i$ to $j$ ending with non-negative SOC is given by the summation at each SOC level. 

Let $\mathcal{I}_{\kappa}$ be an $n(\kappa +1) \times n$ block matrix with $\kappa + 1$ blocks of identity matrix $I$ given by
\begin{equation}
\mathcal{I}_{\kappa} = 
\left[\begin{array}{c}
I \\
I \\
\vdots \\
I
\end{array}\right], \textsf{ and }
\mathcal{Z}_{\kappa}=
\left[\begin{array}{c}
I \\
0 \\
\vdots \\
0
\end{array}\right].
\label{block_identity}
\end{equation}
Then the matrix  $\mathcal{I}_{\kappa}^T\mathcal{B}_{\kappa}^k\mathcal{I}_{\kappa}$, is an $n\times n$ matrix whose $ij$th entry gives the number of feasible walks from $i$ to $j$ of length $k$. 
Let $S$ be the $n(\kappa +1) \times n(\kappa + 1)$ matrix with $ij$ term given by
\begin{equation}
s_{ij} = \sum_{k=1}^{\infty}\alpha^k[\mathcal{B}_{\kappa}^k]_{ij}.
\end{equation}
Thus,
\[
\begin{array}{ll}
S &= I_{n(\kappa + 1) \times n(\kappa +1)} + \alpha \mathcal{B}_{\kappa} + \alpha^2 \mathcal{B}_{\kappa}^2 + \dots +  \alpha^i \mathcal{B}_{\kappa}^i + \dots \\
&= (I_{n(\kappa + 1) \times n(\kappa +1)} - \alpha \mathcal{B}_{\kappa} )^{-1}
\end{array}.
\]
Then, if $W$ is the $n \times n$ matrix given
\begin{equation}
W =  \mathcal{Z}_{\kappa}^T(I_{n(\kappa+1) \times n(\kappa+1)} - \alpha \mathcal{B}_{\kappa} )^{-1}\mathcal{I}_{\kappa}
\end{equation}
The centrality is then given by $C = W \textbf{1}$, where $\textbf{1}$ is the column vector consisting of all 1's.
For the standard centrality measure, the Katz centrality is computed by $(I - \alpha A)^{-1}$. The parameter $\alpha$, also known as the \emph{damping factor}, must be chosen carefully such that $0 < \alpha < 1/\lambda_{\max}$, where $\lambda_{\max}$ is the largest eigenvalue of $\mathcal{B}_{\kappa}$. 

\begin{lemma}
	$1/\lambda_{max}(A) \leq 1/\lambda_{max}({\mathcal{B}_{\kappa}})$ .
	\label{lamda}
\end{lemma}
\begin{proof}
	If $\mathcal{A} = \mathcal{Z}_m^T\mathcal{B}_m \mathcal{I}_m$, then $\mathcal{A}^k$ is the $n \times n$ matrix that counts the walks of length $k$ in the bounded-walk graph. Clearly, $[A^k]_{ij} \geq [\mathcal{A}^k]_{ij}$ for all $i,j$. This implies that if the sequence $\{\alpha^k A^k\}_{k=1}^{\infty}$ converges, then $\{\alpha^k \mathcal{A}^k\}_{k=1}^{\infty}$ converges. Thus, $\lambda_{max}(A) > \lambda_{max}({\mathcal{B}})$ 
\end{proof}

When $\alpha \to 0$, then only walks of very short length are taken into account and degree centrality usually performs well  \cite{borgatti2005centrality,kitsak2010identification}. However, as the value of $\alpha$ increases, eigenvector and Katz outperform other measures \cite{liu2016locating}.  Due to lemma \ref{lamda}, we are able to take larger values of $\alpha$ compared to the standard Katz centrality measure.

\section{Betweenness Centrality}
For a graph $G$, let $\sigma_{st}$ be the total number of shortest paths from nodes $s$ to $t$, while $\sigma_{st}(v)$ be the total number of shortest paths from $s$ to $t$ that pass through $v$. Then the (unnormalized) betweenness centrality of $v$, $BC(v)$, is given by 
\begin{equation}
BC(v) := \sum_{s \neq v \neq t} \frac{\sigma_{st}(v)}{\sigma_{st}}.
\end{equation}
The decision of whether or not to include the end-points of a path to fall on that path is usually made according to specific applications and goals. This is because the only difference this makes is an additive constant to $BC(v)$. In this paper, we will generally include the end-points.

Let  $\sigma_{st}^*$ be the total number of \emph{shortest feasible walks} (see Definition \ref{def:fw}) from  $s$ to $t$, with $\sigma_{st}^*(v)$ be the total number of shortest feasible walks from $s$ to $t$ that pass through $v$. Then the (unnormalized) \soc-betweenness centrality of $v$, $BC^*(v)$, is given by 
\begin{equation}
BC^*(v) := \sum_{s \neq v \neq t} \frac{\sigma_{st}^*(v)}{\sigma_{st}^*}.
\end{equation}
The computation of $BC^*$ depends on counting the number of shortest feasible walks for each pair $s,t \in V$.

\subsection{Counting Shortest Feasible Walks}
Let $d_G(u,v)$ for $u, v
\in V$ be the geodesic distance from $u$ to $v$. The term $\sigma_{st}(v)$ for $v \in V$ can be calculated as
\begin{equation}
\sigma_{st} (v) = \left\{\begin{array}{lr}
0, & \mbox{ if } d_G(s,t) < d_G(s,v) + d_G(v,t)\\
\sigma_{sv}\cdot \sigma_{vt}, & \mbox{otherwise.}
\end{array}\right.
\end{equation}
This property, however, does not hold for counting shortest feasible walks in $G$.  Thus, in order to count feasible walks in $G$, we turn to the directed graph $\mathcal{G}$.

\begin{lemma}
	Let $w = (s=u_0, u_1, \dots, u_k = t)$ be an $s$-$t$ walk in $G$ of length $k$, with $u_i \in V$. $w$ is a \emph{feasible walk} in $G$ if and only if there exists a walk in $\mathcal{G}$ with a node sequence of $(s, \kappa) = (u_0, i_0), (u_1, i_1), \dots, (u_k, i_k) = (t, i_k)$ for some $0 \leq i_0, \dots, i_k \leq \kappa$.
	\label{lemma1}
\end{lemma}
\begin{proof}
	For walk $w=(s=u_0, u_1, \dots, u_k = t)$ in $G$, let $i_0, i_1, \dots, i_k$ be the \soc value at nodes $u_0, \dots, u_k$ respectively during the walk. Since, $w$ is a feasible walk, $i_0 = \kappa$ and node $(u_{j+1}, i_{j +1})$ is adjacent to $(u_{j}, i_{j})$ in $\mathcal{G}$, for $0\leq j \leq k -1$. Therefore the node sequence $(s, \kappa) =$ $(u_0, i_0), (u_1, i_1), \dots, (u_k, i_k) = (t, i_k)$ is a walk in $\mathcal{G}$.
	On the other hand, if $(s, \kappa) =$ $(u_0, i_0), (u_1, i_1)$ ~$, \dots, (u_k, i_k)$ $ = (t, i_k)$ is a walk in $\mathcal{G}$, then $u_{j+1}$ is adjacent to $u_{j}$ in $G$, for $0\leq j \leq k -1$.  Thus, $w=(s=u_0, u_1, \dots, u_k = t)$ is a walk in $G$. 
\end{proof}
For $0\leq i \leq \kappa$, a walk from $(s, \kappa) \in \mathcal{V}$ to $(t, i)$, can be viewed as a feasible walk from $s$ to $t$ in $G$. However, a shortest path from $(s, \kappa) \in \mathcal{V}$ to $(t, i)$ is not necessarily a shortest feasible walk from $s$ to $t$ in $G$. In order to count shortest feasible walks in $G$, we introduce a set of dummy nodes into $\mathcal{G}$ and call the new graph $\mathcal{G}_{\star}  = (\mathcal{V}_{\star}, \mathcal{E}_{\star})$ where
\[
\begin{array}{rl}
\mathcal{V}_{\star} &:= \mathcal{V} \cup \{\ (u , \star) \ | \ u \in V \} \\
\mathcal{E}_{\star} &:= \mathcal{E} \cup \{\ \big((u, i), (u, \star)\big) \ | \ (u,i) \in \mathcal{E} \}
\end{array}
\]
Note: $\star$ can be viewed as a string or marker and is not a variable. The nodes $(u,i)$ for $0 \leq i \leq \kappa$ represent node $u\in V$ at different states $i$. However, for a shortest feasible walk from $s$ to $t$, we are interested in arriving at $t$ at any state, thus introducing a dummy node $(t, \star)$ to capture all final states. The following lemma shows how adding these dummy nodes, simplifies the process representing shortest feasible walks. 
\begin{lemma}
	Let $w = (s=u_0, u_1, \dots, u_k = t)$ be an $s$-$t$ walk in $G$ of length $k$, with $u_i \in V$. $w$ is a \emph{shortest} feasible walk in $G$ if and only if there exists a \emph{shortest path}, in $\mathcal{G}_{\star}$ with a node sequence of $(s, \kappa) = (u_0, i_0), (u_1, i_1), \dots, (u_k, i_k)=(t, i_k), (u_{k+1}, i_{k+1})= (t, \star)$ for some $0 \leq i_0, \dots, i_k \leq \kappa; k \in \NN$. 
	\label{lemma2}
\end{lemma}
\begin{proof}
	If $w$ is a shortest feasible walk in $G$, from Lemma \ref{lemma1}, it follows that there exists a walk $w'$ in $\mathcal{G}$ and subsequently in $\mathcal{G}_{\star}$ with $w' = ((s, \kappa) = (u_0, i_0), (u_1, i_1), \dots,(u_{k}, i_{k}) = (t, i_k))$. Suppose $w'$ is not a path in $\mathcal{G}$. Then there exists a node $(u_j, i_j)$ for some $j \in \NN$ visited more than once. This forms a cycle $C$ within $w'$. Define $w''$ as a node sequence in $\mathcal{G}$ where the cycle $C$ in $w'$ is replaced with $(u_j, i_j)$. It is easy to see that $w''$ is a walk in $\mathcal{G}$ with length strictly less than $w'$. By Lemma \ref{lemma1}, this implies that there exists a feasible walk in $G$ with length less than $w$ contradicting the assumption that $w$ is a shortest feasible walk in $G$. Thus, $w'$ and subsequently the node sequence $(s, \kappa) = (u_0, i_0), (u_1, i_1), \dots, (u_k, i_k)=(t, i_k), (u_{k+1}, i_{k+1})= (t, \star)$ for some $0 \leq i_0, \dots, i_k \leq \kappa$, form a shortest path in $\mathcal{G}_{\star}$.
	
	On the other hand, suppose $w'$ is a shortest path from $(s, \kappa)$ to $(t, \star)$ in $\mathcal{G}_{\star}$ of length $k +1$ for some $k \in \NN$. Lemma \ref{lemma1}, implies that, there exists an $s$-$t$ feasible walk $w$ in $G$ of length $k$. By the same lemma, the existence of a shorter $s$-$t$ feasible walk in $G$ would imply the existence of a walk from $(s, \kappa)$ to $(t, \star)$ in $\mathcal{G}_{\star}$ with length less than $k +1$, contradicting the shortest path assumption.
\end{proof}
For a set $A \subset V$ and $s, t\in V$, define $\sigma_{st}(A)$ as the number of $s$-$t$ shortest paths that pass through one or more nodes in $A$.
\begin{lemma}
	For $s, t, v \in V$, let $\gamma, \tau \in \mathcal{E}_{\star}$ with $\gamma = (s, \kappa)$, $ \tau = (t, \star)$ and $A_v = \{ (v, i)| 0 \leq i \leq \kappa\}$then:
	\begin{enumerate}
		\item $\sigma_{st}^* = \sigma_{\gamma \tau}$
		\item $\sigma_{st}^*(v) = \sigma_{\gamma \tau}\big(A_v\big)$
	\end{enumerate}
\end{lemma}
\begin{proof}
	For the first part, Lemma \ref{lemma2} gives a one-to-one mapping between set of shortest feasible walks in $G$ and set of shortest paths in $\mathcal{G}_{\star}$ whose cardinalities are given by $\sigma_{st}^*$ and $\sigma_{\gamma \tau}$ respectively. For the second part, $\sigma_{\gamma \tau}(A_v)$ counts the number of shortest paths in $\mathcal{G}_{\star}$ that pass through a node in $A_v$. Applying lemma \ref{lemma2}, any shortest path that passes through a node in $A_v$ is a shortest feasible walk in $G$ that passes through $v$. 
\end{proof}
Due to the above lemma, we can compute $BC^*(v)$ as follows: 
\begin{theorem}
	For graph $G=(V,E)$, and directed graph $\mathcal{G}_{\star}$, let $S = \{(s, i)\in \mathcal{E}_{\star} | s \in V, i = \kappa \}$ and $T = \{(t, i)\in \mathcal{E}_{\star} | s \in V, i = \star \}$, $A_v = \{(v, i) | 0 \leq i \leq \kappa\}$, with $v \in V$, then
	\begin{equation}
	BC^*(v) = \sum_{\gamma, \tau \in \mathcal{E}_{\star}\\ \gamma \in S, \tau \in T}\frac{\sigma_{\gamma \tau}\big(A_v\big)}{\sigma_{\gamma \tau}}.
	\label{mod_bc}
	\end{equation}
\end{theorem}
We now show how to compute $BC^*(v)$ without explicitly constructing $\mathcal{G}_{\star}$.  In our computations, the value $\sigma_{\gamma \tau}\big(A_v\big)$ is approximated by $\displaystyle \sum_{i=0}^{\kappa} \sigma_{\gamma \tau}\big((v, i)\big)$.

\subsection{Computing SOC-Betweenness Centrality}
In order to compute $BC^*$, we build on Brandes' algorithm \cite{brandes2001faster} for computing $BC(v)$. We first give a summary of Brandes' algorithm.

The \emph{pair-dependency} is defined as the ratio 
\begin{equation}
\delta_{st}(v) := \frac{\sigma_{st}(v)}{\sigma_{st}}
\end{equation} of a pair $s, t \in V$ on an intermediary node $v\in V$. In order to eliminate the need for explicit summation of all pair-dependencies, Brandes introduces the notion of \emph{dependency} of a vertex $s\in V$ on  a single vertex $v \in V$, defined as 
\begin{equation}
\delta_{s \bullet}(v) := \sum_{t \in V}\delta_{st}(v)
\end{equation}
and shows that the dependency of $s \in V$ on any $v \in V$ obeys the following recursive relation
\begin{equation}
\delta_{s \bullet}(v) = \sum_{w: v \in P_s(w)} \frac{\sigma_{sv}}{\sigma_{sw}} \cdot (1 + \delta_{s \bullet}(w)).
\label{standard_recur}
\end{equation}
where $P_s(v)$ is the set of \emph{predecessors} of a vertex $v$ during a breadth-first search (BFS) from source $s\in V$. It is given by
\begin{equation}
P_s(v) := \{u \in V: \{u,v\} \in E, \ d_G(s,v) = d_G(s,u)  + 1 \}
\end{equation}
where $d_G(s,v)$ is the geodesic distance from $s$ to $v$. In summary, Brandes' algorithm for computing BC is as follows: for each source node, $s \in V$,
\begin{enumerate}[i]
	\item perform BFS computing number of shortest paths to every other node, $t \in V$
	\item back propagation: compute $\delta_{s\bullet}(v)$ for $v \in V$ in order of non-increasing distance from $s$.
\end{enumerate}
One major difference between equation (\ref{mod_bc}) and the standard computation of $BC$ is that equation (\ref{mod_bc}) is constrained by the fact that the source and target nodes must be chosen from sets $S$ and $T$. Therefore, the recursive relation given by equation (\ref{standard_recur}) can not be used as it is because not all nodes are target nodes.
For $T \subset V$, consider the function
\begin{equation}
\delta_{s \bullet}^T(v) := \sum_{t \in T}\delta_{st}(v),
\label{target}
\end{equation}
then 
\begin{lemma}
	\begin{equation}
	\delta_{s \bullet}^T(v) = \sum_{w: v \in P_s(w)} \frac{\sigma_{sv}}{\sigma_{sw}} \cdot (\mathbbm{1}_{T}(w) + \delta_{s \bullet}^T(w)),
	\label{recur_target}
	\end{equation}
\end{lemma}
\noindent where $\mathbbm{1}_T(w)$ is the indicator function such that $\mathbbm{1}_T(w) = 1$ if $w \in T$ and 0 otherwise. 
\begin{proof}
	For each term on the right side, if $w \in T$, then the summand follows from  equation  (\ref{standard_recur}).\\ Consider the case if $w \notin T$.
	Extend the definition of the pair-dependency to include an edge $e$ such that, $\delta_{st}(v,e):= \sigma_{st}(v,e)/\sigma_{st}$ where $\sigma_{st}(v,e)$ is the number of shortest $s$-$t$ paths that contain both $v$ and $e$. Then Brandes showed that
	
	\begin{equation*}
	\delta_{s \bullet}(v) = \displaystyle \sum_{w: v \in P_s(w)} \sum_{t \in V}\delta_{st} (v, \{v, w\}) 
	\end{equation*}
	and
	\begin{equation*}
	\delta_{st} (v, \{v, w\}) = \left\{\begin{array}{lr}
	\frac{\sigma_{sv}}{\sigma_{sw}}, & \mbox{ if } t = w\\
	\frac{\sigma_{sv}}{\sigma_{sw}}\cdot \frac{\sigma_{st}(v)}{\sigma_{st}}, & \mbox{ otherwise}
	\end{array}\right.
	\end{equation*}	
	It then follows that 
	\begin{equation*}
	\delta_{s \bullet}^T(v) = \displaystyle \sum_{w: v \in P_s(w)} \sum_{t \in T}\delta_{st} (v, \{v, w\}) .
	\end{equation*}
	So for $w \notin T$, then $t \neq w$ and
	\begin{equation*}
	\begin{array}{ll}
	\displaystyle \sum_{w: v \in P_s(w)} \sum_{t \in T}\delta_{st} (v, \{v, w\}) &= 
	\displaystyle \sum_{w: v \in P_s(w)}\frac{\sigma_{sv}}{\sigma_{sw}}\cdot \frac{\sigma_{st}(w)}{\sigma_{st}}\\
	& = \displaystyle \sum_{w: v \in P_s(w)}\frac{\sigma_{sv}}{\sigma_{sw}}\cdot \delta{s \bullet}^T(w)
	\end{array}
	\end{equation*}	
\end{proof}
The recursive relation in equation \ref{recur_target} is used to compute the SOC-betweenness centrality. Algorithm (\ref{soc_btn}) describes this computation in detail.

\begin{algorithm}[ht]
	\caption{SOC-Betweenness Centrality}\label{soc_btn}
	\begin{algorithmic}[1]
		\State \textbf{Input:} $G=(V,E),  \mathcal{V}_{\star}, \Omega, \kappa$
		\State \textbf{Output:} $bc[v], v \in \mathcal{V}_{\star}$
		\State $bc[\nu] \gets 0, \nu \in \mathcal{V}_{\star}$
		\State $\Sigma \gets \{(u, i) \in \mathcal{V}_{\star} | u \in \mathcal{V}_{\star} , i = \kappa\}$
		\For {$s \in \Sigma$}
		\State $S \gets$ empty stack;
		\State $P[\omega] \gets$ empty list, $\omega \in \mathcal{V}_{\star} $;
		\State $\sigma[t] \gets 0, t \in \mathcal{V}_{\star} ; \sigma[s] \gets 1$
		\State $d[t] \gets -1, t \in \mathcal{V}_{\star} ;\ d[s] \gets 0$
		\State $Q \gets $ empty queue;
		\State enqueue $s \to Q;$
		\While {$Q \ not \ empty$}
		\State dequeue $(v,i) \gets Q;$
		\State push $(v,i) \to S;$
		\If {$i \neq \star$}
		\For {$neighbor \ w \ of v$}
		\If {$w \in \Omega$} $current\_node \gets (w, \kappa)$
		\Else 
		\If {$i \geq 0$}  $current\_node \gets (w, i-1)$
		\Else 
		$\ current\_node \gets -1$
		\EndIf
		\EndIf
		\If {$current\_node \neq -1$}
		\If {$d[current\_node] < 0$}
		\Comment $w \ found \ for \ first \ time ?$
		\State enqueue $current\_node \to Q;$
		\State $d[current\_node] \gets d[(v,i)] + 1;$
		\EndIf
		\If {$d[current\_node] = d[(v,i)] + 1$}
		\Comment{shortest path to $w$ via $v$?}
		\State $\sigma[current\_node] \gets \sigma[current\_node] + \sigma[(v,i)];$
		\State append $ (v,i) \to P[current\_node];$
		\EndIf 
		\EndIf
		
		\EndFor
		\EndIf
		
		\EndWhile
		\State $\delta[\nu] \gets 0, \nu \in \mathcal{V}_{\star};$
		\Comment{$S$ return vertices in order of non-increasing distance from $s$}
		\State $\chi[(v, i)] \gets 0, (v,i) \in \mathcal{V}_{\star} $
		\State $\chi[(v, i)] \gets 1, (v,\star) \in \mathcal{V}_{\star} $
		\While {$S \ not \ empty$}
		\State pop $(w,i) \gets S;$
		\If {$ \chi[(w,i)] = 1$}
		\For {$(v,j) \in P[(w,i)]$}
		$\chi[(v,j)] \gets 1$
		\If {$i = \star$}
		$\ \delta[(v,j)] \gets \delta[(v,j)] + \frac{\sigma[(v,j)]}{\sigma[(w,i)]}\cdot (1 + \delta[(w,i)]);$
		\Else  
		~$\delta[(v,j)] \gets \delta[(v,j)] + \frac{\sigma[(v,j)]}{\sigma[(w,i)]}\cdot \delta[(w,i)];$
		\EndIf
		\EndFor
		\If {$(w,i) \neq s$}
		$bc[(w,i)] \gets bc[(w,i)] + \delta[(w,i)]$
		\EndIf
		\EndIf
		\EndWhile
		\EndFor
	\end{algorithmic}
\end{algorithm}

\section{Random-Walk Betweenness Centrality}
A common criticism for betweenness centrality is that it does not take non-shortest paths into account and is therefore inappropriate in cases where information spread is governed by other rules \cite{borgatti2005centrality}. As a result, variants of betweenness centrality have been proposed such as betweenness measures based on network flow  \cite{freeman1991centrality}, and random-walk betweenness centrality (RWBC) \cite{newman2005measure,brandes2005centrality}. In some sense, as suggested by Newman,  ``RWBC and BC can be viewed as being on opposite ends of a spectrum of possibilities, one representing information that is moving at random and has no idea of where it is going and the other knowing precisely where it is going''. Some real-world situations mimic these extremes \cite{newman2005measure,freeman1978centrality}, 
however, others such as the small-world experiment \cite{kleinfeld2002small} fall somewhere in between. 

In a network where the flow process is coupled with a SOC constraint, it is therefore natural to also propose a variant of RWBC for such networks. If we consider an undirected connected graph, for any pair of nodes $s,t$, a random walk starting at $s$ will eventually arrive at $t$ with high probability. However, in a network where the flow is coupled with a SOC constraint, and likewise, a directed network that is not strongly connected, not every random walk starting at $s$ has a positive probability of arriving at $t$.  With this in mind, the proposed variant of RWBC only considers walks that arrive at the destination node. For example, if a node does not have enough SOC to travel from $s$ to $t$ via any walk, then the pair $s$-$t$ does not contribute to centrality score.

Consider RWBC proposed in \cite{newman2005measure}. Unlike the standard betweenness centrality measure that only considers shortest paths between a pair of nodes, RWBC takes all paths into account while giving more importance to shorter paths.  RWBC of a node $i$ is defined as the \emph{net} number of times a random walk passes through $i$. By \emph{net}, authors meant that if a walk passes through $i$ and later passes back through it in the opposite direction, the two would cancel out and there is no contribution to the betweenness. 

RWBC was originally proposed for undirected graphs. In this section, we first generalize RWBC to directed graphs.  In a directed graph $\mathcal{G} = (\mathcal{V}, \mathcal{E})$, for any pair of nodes $s,t \in \mathcal{V}$, it is not guaranteed that every random walk from $s$ will eventually arrive at $t$. We generalize RWBC for directed graphs to only include random walks from $s$ to $t$. Let $\mathcal{G}_{s,t} = (\mathcal{V}_{s,t}, \mathcal{E}_{s,t})$ be a subgraph of $\mathcal{G}$ such that every node lies on a walk from $s$ to $t$. RWBC is adjusted for $\vec{\mathcal{G}_{s,t}}$ as follows. Let $\mathcal{A}$ adjacency matrix with $D$ the out-degree diagonal matrix, where $D$ is defined as
\begin{equation*}
[D]_{ij}  :=  \left\{\begin{array}{ll}
deg_{+}(v_i) & \mbox{ if } i =j\\
0, & \mbox{ otherwise},
\end{array}\right.
\end{equation*}
\noindent where $deg_{+}(v_i)$ is the out-degree of node $v_i$. Define the transition matrix of $\vec{\mathcal{G}_{s,t}} $ as
\begin{equation}
M := D^{-1}\mathcal{A}
\end{equation}

For a walk starting at $s$, the probability that it is at $j$ after $r$ steps is given by $[M^r]_{sj}$. The probability that the walk continues further to an adjacent vertex $i$ is $[M^r]_{sj}d^{-1}_{j}$, where $d^{-1}_{j}$ is the out-degree at $j$. Thus, the expected number of times a walk from $s$ to $t$ uses the directed edge $(j, i)$ is given by $[(I -M_t)^{-1}]_{sj}d^{-1}_{j}$,
which is the $s$-$j$th entry of the matrix given by 
\begin{equation}
(I - M_t)^{-1}D^{-1}_{t} = (D_{t} - A_t )^{-1},
\end{equation}
where $D_t$ and $A_t$ is the matrix derived from deleting row and column $t$. 
Add the zero column back to $(D_{t} - A_t )^{-1}$ and call this matrix $T$. 
Let $\textbf{s}$ be the vector given by 
\begin{equation*}
s_i := \left\{\begin{array}{lr}
1,& \mbox{if } i = s\\
-1,& \mbox{if } i = t\\
0, &\mbox{otherwise}
\end{array}\right.
\end{equation*}
Let the vector $f$ be defined as 
\begin{equation*}
f := \textbf{s}^T T 
\end{equation*}
then, the $i$th entry of $f$, $f_i$, represents the expected number of walks from $s$ to $t$ that pass through node $i$. If $\mathcal{D}_f$ is the diagonal matrix with $f_i$ at the $i$th diagonal position then the matrix
\begin{equation*}
\mathcal{F} := \mathcal{D}_f\mathcal{A}
\end{equation*}
gives a matrix whose $i$-$j$ value represents the expected number of times a random walk from $s$ to $t$ uses edge $(i,j)$. 
The net flow of random walk through the $i$th vertex is for a given $s$-$t$ pair is given by
\begin{equation}
I_i^{(st)} = \frac{1}{2}\sum_{(i,j)\in \mathcal{E}_{s,t}} |\mathcal{F}_{i,j} - \mathcal{F}_{j,i}|.
\label{rwbc}
\end{equation}

The expression in (\ref{rwbc}) is used to compute the centrality scores of a directed graph $\mathcal{G}_{s,t}$ arising from random walks starting at node $s$ to $t$. Note that for every node $v$ in $\mathcal{G}_{s,t}$, $v$ must be at a finite distance from $s$ and $t$. The centrality for each node in $\mathcal{G}$ is then given by the sum of the individual scores for each source-target pair. Let $\hat{Y}$ be the vector of centrality scores of $\mathcal{G}$, the \soc-RWBC is given by the vector
\begin{equation}
\hat{Y}^T\cdot \mathcal{I}_{\kappa}
\end{equation}
where $\mathcal{I}_{\kappa}$ is the block matrix defined in (\ref{block_identity}).

\section{Computational Experiments}
In the preceding sections, mathematical models for the three proposed centrality measures are given. The question then arises, ``How good are these centrality measures?'' We tackle this question from three different perspectives. First, \emph{usability}: how can we meaningfully use the proposed centrality measures. Second, \emph{robustness}: how robust are the centrality measures with respect to their parameters. Lastly, \emph{novelty}: how are the proposed centrality measures different from their well-established predecessors. Experiments in this section are carried on the graph datasets given in Table \ref{tab:graphs}.

\begin{table}[htp]
	\caption{
		Experimental graph datasets: $d_{\min}, d_{\mbox{avg}}, d_{\max}$ represents the minimum, average and maximum degree.
	}
	\label{tab:graphs}
	\begin{tabular}{p{3.8cm}p{1.2cm}p{1.2cm}p{0.9cm}p{0.9cm}p{0.9cm}p{3.5cm}}
		\hline
		Graph   & Nodes     & Edges & $d_{\min}$ & $d_{\mbox{avg}}$  & $d_{\max}$ &Reference \\\hline
		Router Network  & 2114 & 6632 &1 &6 & 109 & \cite{rossi2015network}\\
		Minnesota State Road Network & 2642& 3303&1&2 &5 &\cite{davis2011university}\\
		Gnutella Network& 6301 & 2077&1&7 &97 & \cite{ripeanu2002mapping}\\
		Collaboration Network & 5242 &14496&0&5&81 & \cite{leskovec2007graph}
		\\  \hline
	\end{tabular}
\end{table}

\subsection{Usability}

In Borgatti \cite{borgatti2005centrality}, the expected centrality is defined as a centrality score given by a closed-form expression, 
and \emph{realized centrality} as the actual centrality score observed in the context of a particular flow process. Therefore, one can view a centrality measure as a formula-based prediction of a flow process through a node. It is therefore important to compare the predictions given by the closed-form expression with the actual frequency of traffic observed flowing through a node across multiple instances. For example, in order to test whether betweenness centrality is a good prediction of observed traffic through a node, the expected betweenness centrality is compared with the realized betweenness centrality, where the actual frequency of traffic through a given node is referred to as the realized betweenness centrality while the formula-based centrality is referred to as the expected betweenness centrality. 
In this section, we compare the expected centrality with the realized centrality values for the proposed centrality measures. Given that realized centrality scores are achieved by running long simulations, we show the usability of the proposed measures as way to efficiently estimate the outcome of these computationally expensive simulations.

In order to observe realized centrality values, simulations for each of the three flow process are developed. In general, centrality measures are primarily used either as ranking algorithms or as methods for identification of influential nodes. We therefore compare the expected and realized centrality values using Kendall's Tau \cite{kendall1938new} 
rank correlation coefficient. Kendall's Tau is given by
\begin{equation*}
\tau := \frac{2}{n(n-1)}\sum_{i<j}sgn \Big[(y_i-y_j)(z_i-z_j)\Big],
\end{equation*}
where, for each node $i$, we denote the node's spreading influence and its  centrality measure by $y_i$ and $z_i$, respectively. The $sgn(y)$ is a piecewise function such that $sgn(y) = 1$ if $y>0$, $-1$ if $y<0$ and 0 if $y =0$. The values of $\tau$ belong to the range $[-1,1]$, where larger values of $\tau$ correspond to a higher correlation between the expected and realized centralities.   

\subsubsection{\soc Katz Centrality}
In order to compute the realized centrality values with respect to \soc-Katz centrality, we turn to the susceptible-infected-recovered (SIR) spreading model (also called susceptible-infected-removed
model) \cite{hethcote2000mathematics}. Klemm et. al \cite{klemm2012measure} suggested that the eigenvector centrality can be used for estimating a spreading influence of the nodes in the SIR model, by \cite{liu2016locating} defining the dynamical-sensitive (DS) centrality and showing that it more accurately locates influential nodes in the SIR model. The DS centrality is very closely related to the Katz centrality.

In the SIR model, a node can be in one of following states: (i) susceptible, nodes can become infected, (ii) infected, nodes are infected and can infect susceptible nodes, and (iii) recovered, nodes have recovered and developed immunity, thus cannot be infected again. In order to estimate the spreading influence of a node $v$, initially, all nodes are susceptible and $v$ is infected. At each step, an infected node tries to infect its susceptible neighbors and succeeds with probability $\alpha$. The infected node enters the recovered state with probability $\mu$. In this work, we set $\mu =1$, i.e., the infection can be transmitted only once. The process stops if no new infections are formed or after a fixed number of steps. We generalize the SIR spreading process to accommodate the \soc parameter. 

Define an edge $(i,j)$ as \emph{active} if node $i$ infected $j$ via edge $(i,j)$. A stopping criteria for the SIR spreading process for a fixed number of steps $\kappa$ can be viewed as follows: Let $u$ be the initially infected node, then
\begin{itemize}
	\item [] \textit{an infected node $v$ cannot infect its susceptible neighbors if there exists a path 
		of length $\kappa$ from $u$ to $v$ consisting of only active edges}, i.e., the infection dies out after $\kappa$ steps.
\end{itemize}
In order to generalize the SIR model to accommodate a flow process based on \soc, we modify the above stopping criteria to:
\begin{itemize}
	\item [] \textit{an infected node $v$ cannot infect its susceptible neighbors if there exists a path of length $\kappa$ from \underline{either} $u$, \underline{or non-susceptible $w\in \Omega$} to $v$ consisting of only active edges.}
\end{itemize}
If the set $\Omega \subset V$, is empty, then \soc-Katz centrality is equivalent to the DS centrality which is shown in \cite{liu2016locating} to be highly correlated to the nodes' spreading influence according to the SIR model.  

Experiments are carried out to show that the above generalized SIR spreading process is highly correlated to the proposed \soc-Katz centrality. For this experiment, we use the 
network representing the Internet at the major router level \cite{rossi2015network, rocketfuel} consisting of 2114 nodes and 6632 edges. The nodes and edges represent routers, and the connections between them, respectively. We set $\kappa = 5$, and $\alpha = 0.03$, and vary the size of the set $\Omega \subset V$ such that the ratio $|\Omega|/|V|$ ranges from 0.1 to 0.9. For each value of $|\Omega|/|V|$, the set $\Omega$ is chosen at random, and the corresponding spreading influence is estimated for each node by running the generalized SIR model $10^4$ times. This is repeated 30 times. The box-plot in Figure \ref{katz_boxplot} shows the correlation between the spreading influence as a result of the generalized SIR model compared to \soc-Katz centrality. The results show Kendall Tau correlation values in the range $(0.945, 0.970)$ suggesting that the two processes are very highly correlated. 
\begin{figure}[htb]
	\centering
	\begin{minipage}{0.7\textwidth}
		\begin{tikzpicture}
		\node (img)  {\includegraphics[width=1\textwidth]{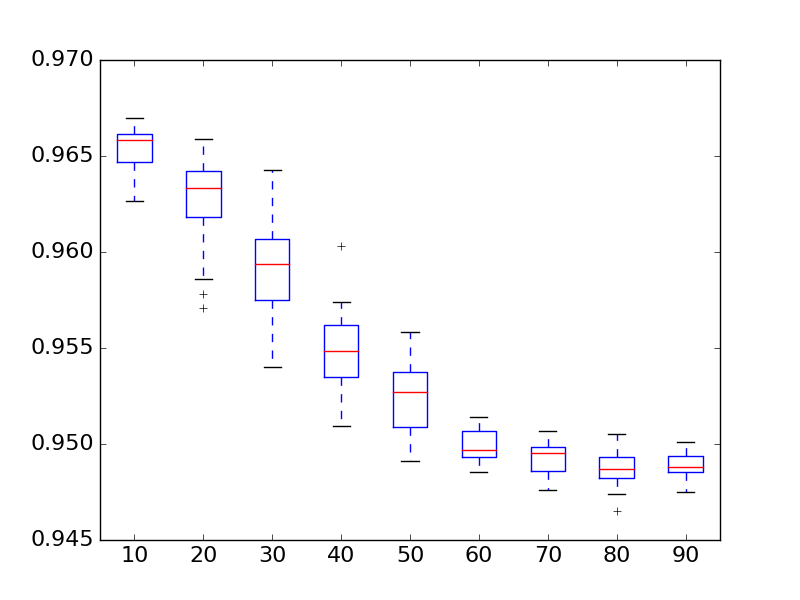}};
		\node[below=of img, node distance=0cm, yshift=1.3cm,font=\color{black}] {$|\Omega|/|V|$  (\%)};
		\node[left=of img, node distance=0cm, xshift=0.5cm, rotate=90, anchor=center,yshift=-0.6cm,font=\color{black}] {Kendall $\tau$};
		\end{tikzpicture}
	\end{minipage}%

	\caption{Comparison of nodes' spreading influence according to the generalized SIR model and \soc Katz centrality 
		on the routers network. Each boxplot represents 30 random choices of the set $\Omega$ with spreading probability $\alpha = 0.03$, and $\kappa = 5$ }
	\label{katz_boxplot}
\end{figure}
\subsubsection{\soc-(Random-Walk) Betweenness Centrality}
To demonstrate the \soc-RWBC and \soc-betweenness centralities, we experiment with two networks, namely, a computer network and road network. The computer network is generated from the P2P network, Gnutella \cite{leskovec2007graph,ripeanu2002mapping}, and consists of 6301 nodes and 20777 edges. The road network  \cite{davis2011university} represents Minnesota state roads and consists of 2642 nodes and 3303 edges. We simulate traffic on both networks. 


The realized centralities are computed using the particle hopping. The particle hopping model is a method used in vehicular flow theory \cite{nagel1996particle}. In this model, a section of a road is represented by a node and a vehicle as a particle where each node can only be occupied by one particle at a given time. This model is sometimes referred to as cellular automata and gives a minimal model for traffic flow behaviors \cite{huitema2000routing}. The flow of packets through the internet have also been modelled by cellular automata \cite{liu2002simple, huisinga2001microscopic}. In \cite{holme2003congestion}, the fraction of time steps that node is occupied by a particle is referred to as the \emph{occupation ratio}.

For the application to electric vehicles, the value $\kappa$ represents the number of steps the car can travel before its battery runs out of charge. 
For a message or vehicle being propagated from node $s$ to $t$, we simulate the traffic on the nodes when a routing algorithm propagates the message or vehicle in one of the two cases, (i) via a shortest feasible walk, and (ii) a random feasible walk. We add the condition that the routing algorithm is informed and takes the current $\kappa$ counter of the message or vehicle, and target $t$, into account before deciding 
which neighbor to direct it to. In other words, if a message or vehicle cannot be successfully propagated to its destination due to the value of $\kappa$, then the message or vehicle is not propagated at all and therefore does not contribute to the traffic of the network. 

Experiments are carried on the Gnutella network, where we set $TTL = 4$ and record the occupation ratio based on the corresponding routing algorithms. The occupation ratio is then compared to the proposed centrality measures. The results for \soc-betweenness centrality and \soc-RWBC are presented in Figure \ref{soc_btn_boxplot} and \ref{soc_rwbc_boxplot} 
respectively.  The results show Kendall Tau values in the range $(0.79, 0.82)$ for a ratio $|\Omega|/|V|$ of 0.2 and $(0.86, 0.88)$ for a ratio $|\Omega|/|V|$ of 0.9 for \soc-betweenness centrality. Similar correlation scores and trends for \soc-RWBC  are observed suggesting a high correlation between the expected centralities and realized centrality measures. 

\begin{figure}[htb]
	\centering
	\begin{minipage}{0.7\textwidth}
		\begin{tikzpicture}
		\node (img)  {\includegraphics[width=1\textwidth]{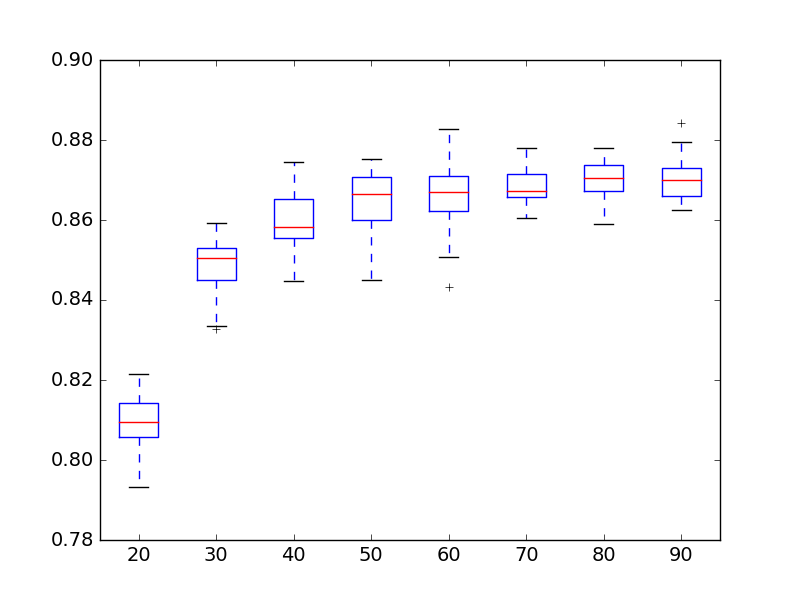}};
		\node[below=of img, node distance=0cm, yshift=1.3cm,font=\color{black}] {$|\Omega|/|V|$  (\%)};
		\node[left=of img, node distance=0cm, xshift=0.5cm, rotate=90, anchor=center,yshift=-0.7cm,font=\color{black}] {Kendall $\tau$};
		\end{tikzpicture}
	\end{minipage}
	\caption{Correlation scores for expected versus realized centrality for \soc-betweenness centrality for the Gnutella network. Each boxplot represents 30 random choices of the set $\Omega$, with $\kappa = 4$}
	\label{soc_btn_boxplot}
\end{figure}
\begin{figure}[htb]
		\centering
	\begin{minipage}{0.7\textwidth}
		\begin{tikzpicture}
		\node (img)  {\includegraphics[width=1\textwidth]{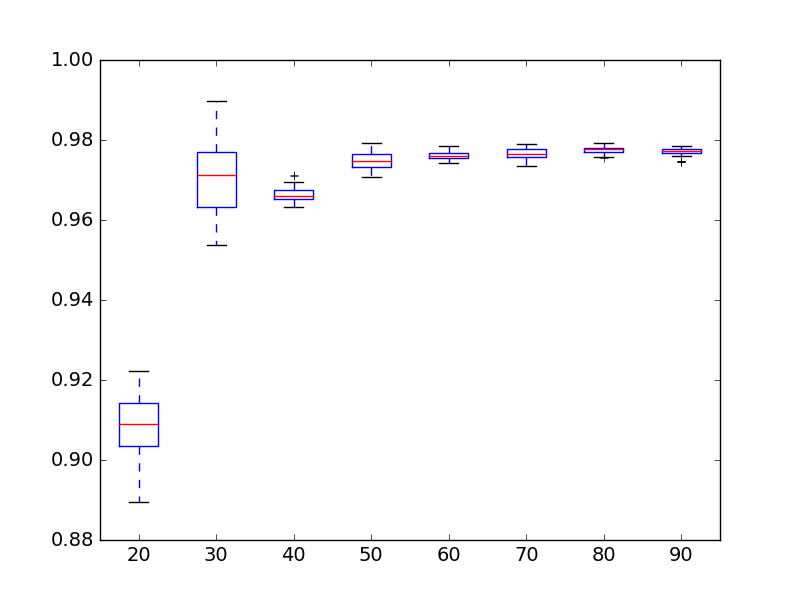}};
		\node[below=of img, node distance=0cm, yshift=1.3cm,font=\color{black}] {$|\Omega|/|V|$  (\%)};
		\node[left=of img, node distance=0cm, xshift=0.5cm, rotate=90, anchor=center,yshift=-0.7cm,font=\color{black}] {Kendall $\tau$};
		\end{tikzpicture}
	\end{minipage}%
	\caption{Correlation scores for expected versus realized centrality for \soc RWBC for the Gnutella network. Each boxplot represents 30 random choices of the set $\Omega$, with $\kappa = 4$}
	\label{soc_rwbc_boxplot}
\end{figure}

For experiments on the Minnesota state road network, we set $\kappa = 20$. We choose a relatively larger value of $\kappa$ for the road network experiments because we assume that electric vehicles can travel a relatively long distance if it starts fully charged. As in the Gnutella experiments, we record the occupation ratio. The results for \soc-betweenness centrality  are presented in Figure \ref{soc_btn_MNboxplot}.  The results show Kendall Tau values in the range $(0.79, 0.86)$ for a ratio $|\Omega|/|V|=$ 0.2 and $(0.83, 0.87)$ for a ratio $|\Omega|/|V|=$ 0.9 for \soc-betweenness centrality.

\begin{figure}[htb]
	\centering
	\begin{tikzpicture}
	\node (img)  {\includegraphics[width=0.7\textwidth]{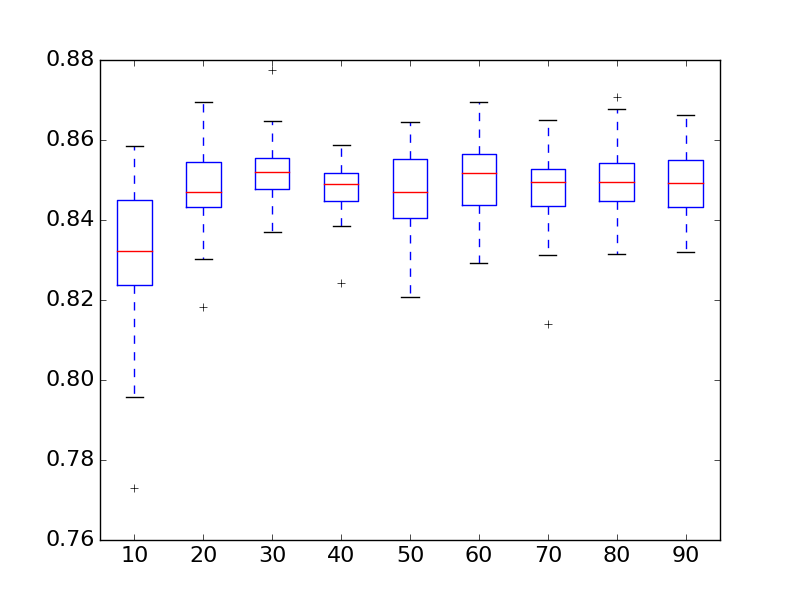}};
	\node[below=of img, node distance=0cm, yshift=1.3cm,font=\color{black}] {$|\Omega|/|V|$  (\%)};
	\node[left=of img, node distance=0cm, xshift=0.5cm, rotate=90, anchor=center,yshift=-0.7cm,font=\color{black}] {Kendall $\tau$};
	\end{tikzpicture}
	\caption{Correlation scores for expected versus realized centrality for \soc-betweenness centrality for the Minnesota state road network. Each boxplot represents 30 random choices of the set $\Omega$, with $\kappa = 20$}
	\label{soc_btn_MNboxplot}
\end{figure}

\subsection{Robustness and Novelty}
The parameter $\kappa$ is application dependent, so it is important to understand how the proposed centrality measures behave for different values of $\kappa$. From the mathematical expressions of our novel  centrality measures, it is clear that for a large enough $\kappa$, the proposed centrality measures would become identical to  their baseline centrality measures as in this case, the limitation of \soc-dependent distance is gradually vanishing. In this section, we carry out experiments to understand how the proposed measures compare to their baseline measures while varying the parameter $\kappa$. The goal of the experiments is to quantify what is not captured when using the well-established centrality measures for given values of $\kappa$. Thus, demonstrating the robustness of the results and novelty of the measures. 

The first set of experiments is carried out on toy graphs to illustrate the difference in central nodes when using the proposed centrality measures versus their baseline measures. The first toy graph is a graph formed by connecting two $5\times5$ grid graphs by a path of length 5. The second is a $10 \times 10$ grid graph. The second set of experiments uses real-world datasets. In these experiments, we focus on the Minnesota state road network and a collaboration network constructed using the scientific collaboration data \cite{leskovec2007graph}, consisting of 5242 nodes and 14496 edges.

The experiments on toy graphs are used to visually illustrate to the reader the difference between the proposed centrality measures versus  their baseline measures. Thus, providing an intuition on how the measures work. The difference between \soc-BC and BC for small values of $\kappa$ is illustrated in Figure \ref{compBC}. In this experiment, we set $\kappa = 4$. We use a color spectrum from red to yellow, showing the most central to the least central nodes respectively. The graph in Figure \ref{compBC} (a) represents the standard BC. As expected, the nodes along the bridge are the most central nodes. The graphs in (b) - (f) show different scenarios where the nodes in $\Omega$ are marked with a blue-edge diamond-shaped node. As we can see in the Figures (d) and (f), depending on the value of $\kappa$ and nodes in $\Omega$, the centrality scores can be significantly different from the standard BC, where the most central nodes based on BC, are now among the least central nodes based on \soc-BC.

\begin{figure}[htp]
	\begin{center}
		\begin{tabular}{cc}
			
			\includegraphics[width=0.5\linewidth]{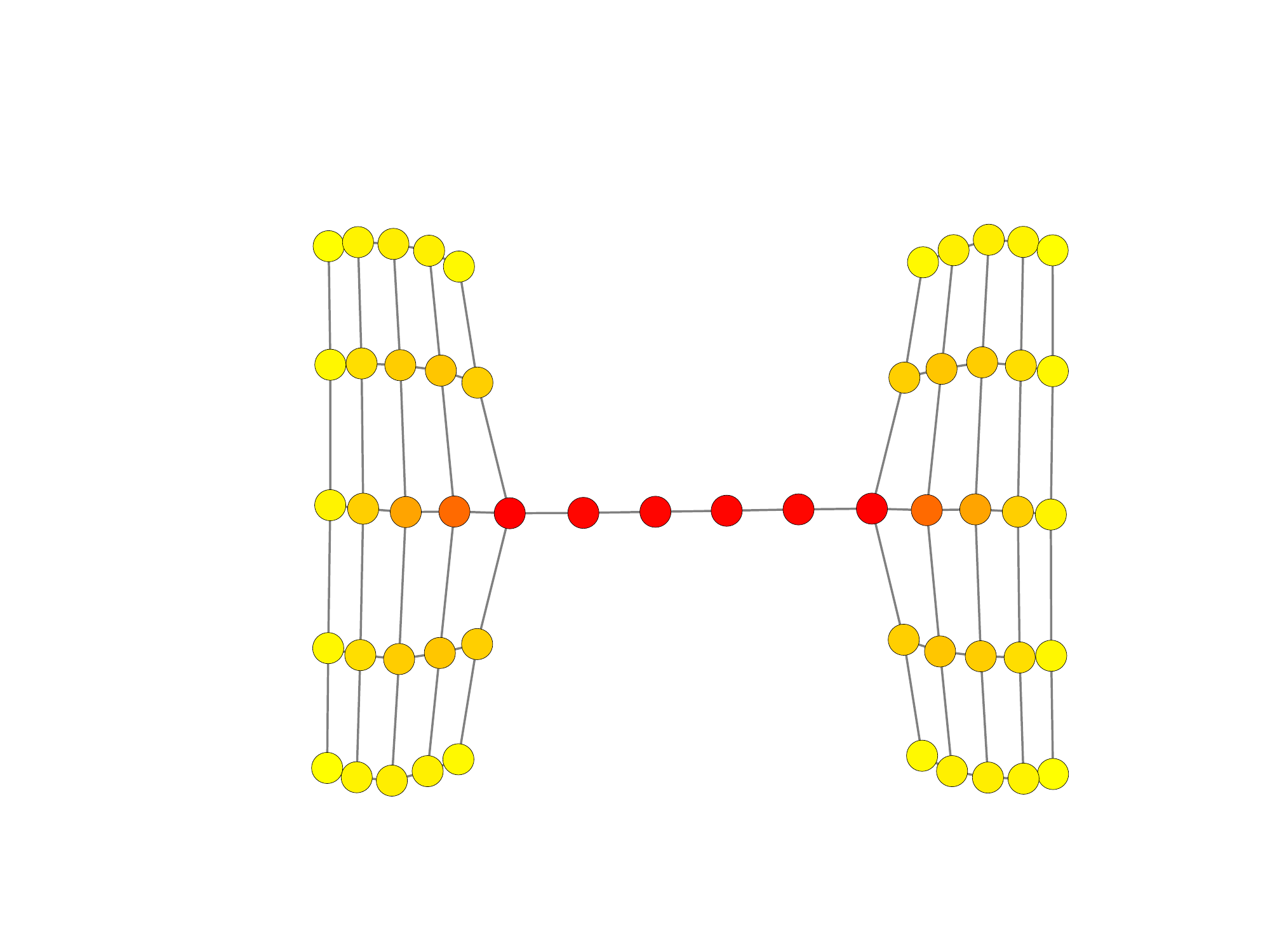} &     \includegraphics[width=0.5\linewidth]{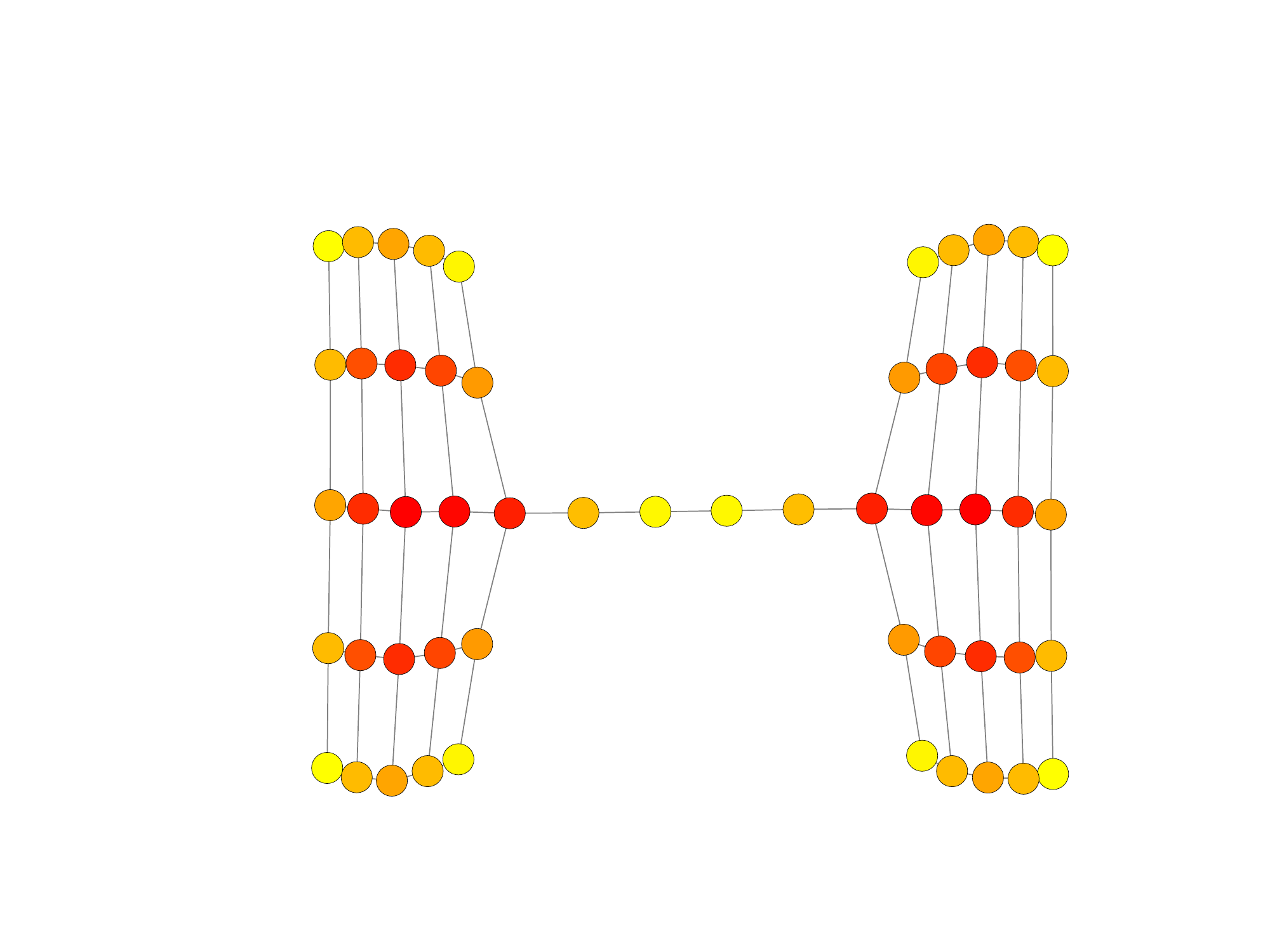} \\
			(a) Standard BC &(b)  $|\Omega| = 0, \kappa = 4$ \\ 
			\includegraphics[width=0.5\linewidth]{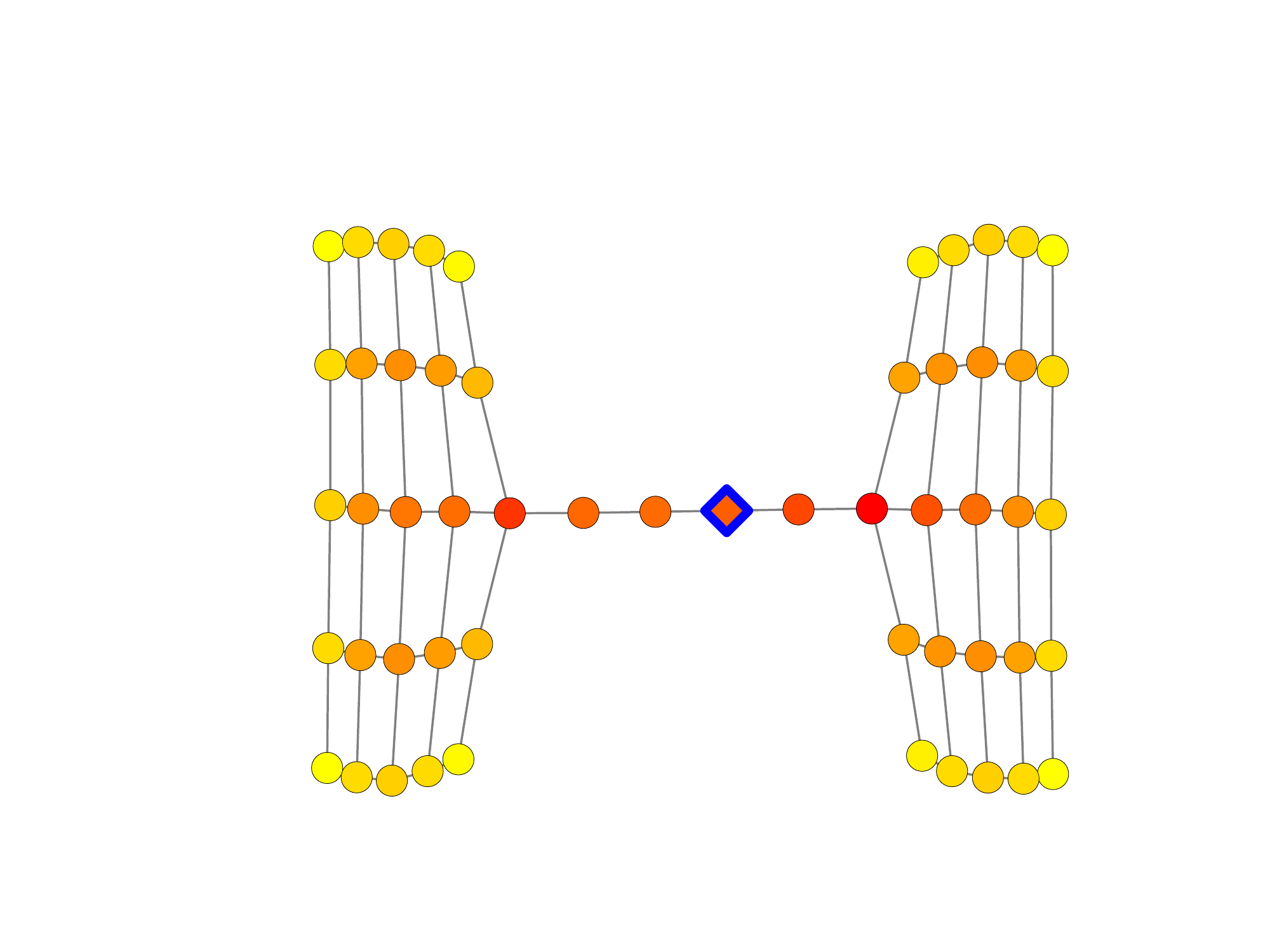}  &\includegraphics[width=0.5\linewidth]{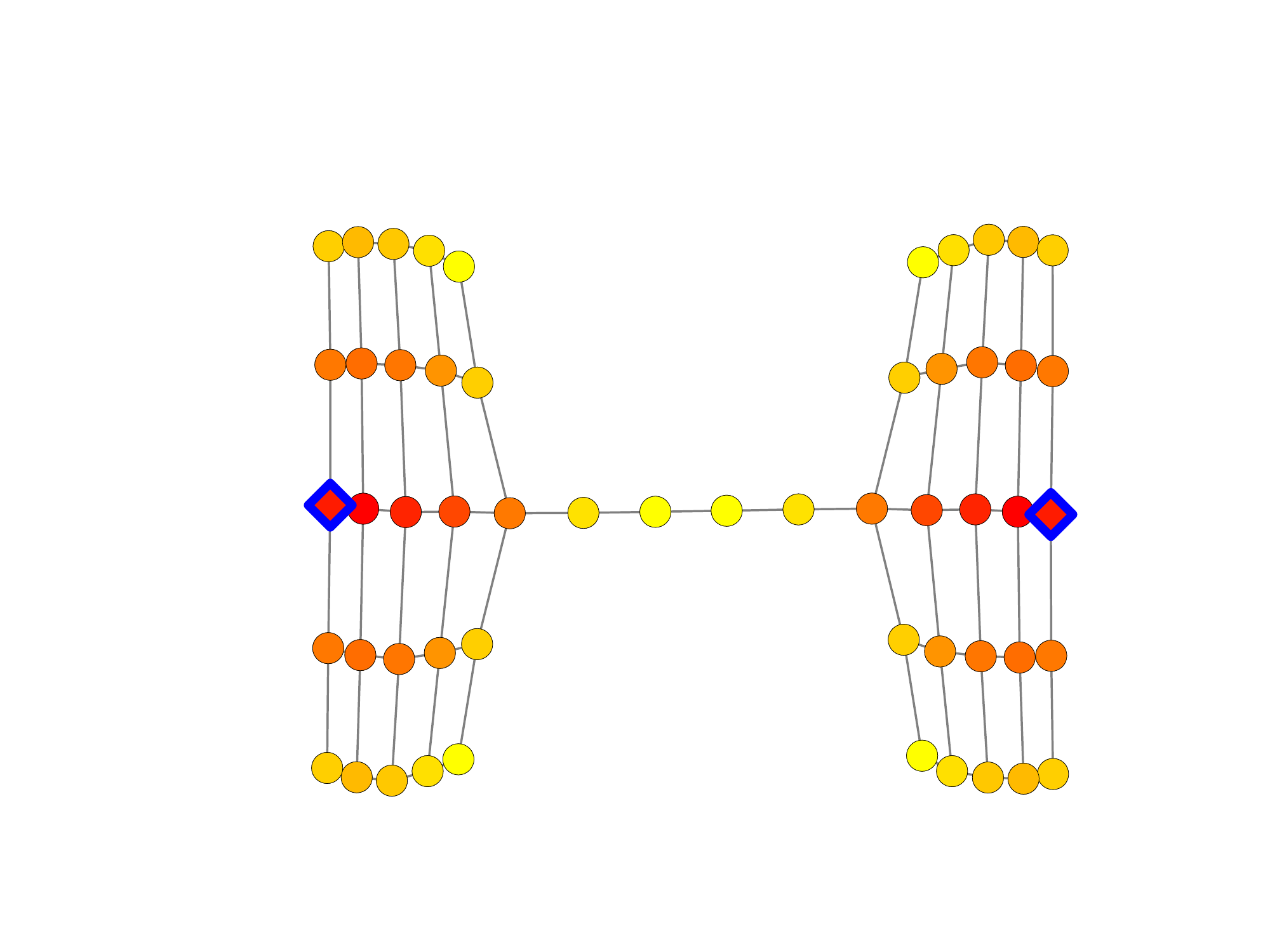}  \\
			(c) $|\Omega| = 1, \kappa = 4$ &(d)  $|\Omega| = 2, \kappa = 4$ \\ 
			\includegraphics[width=0.5\linewidth]{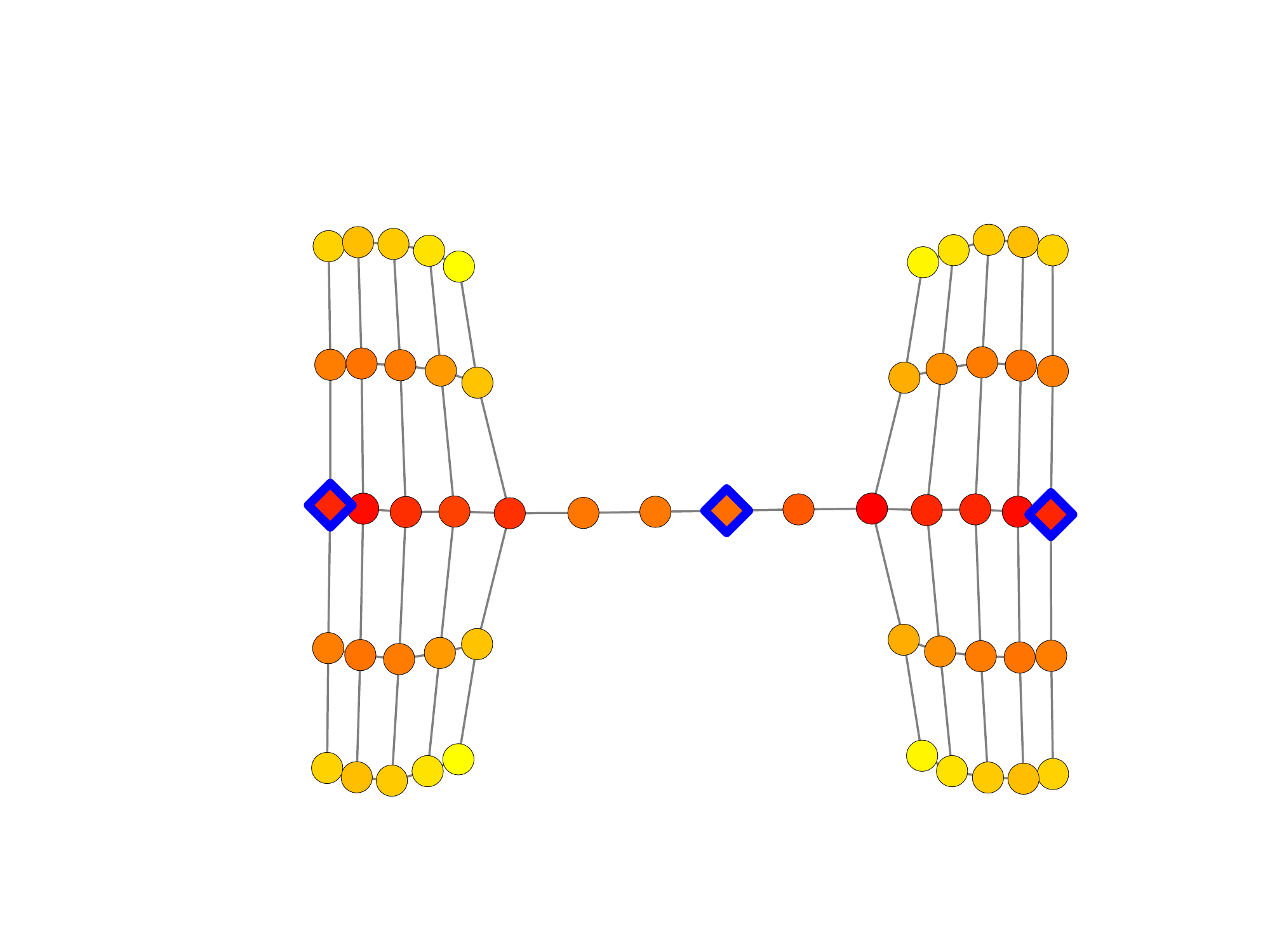}&\includegraphics[width=0.5\linewidth]{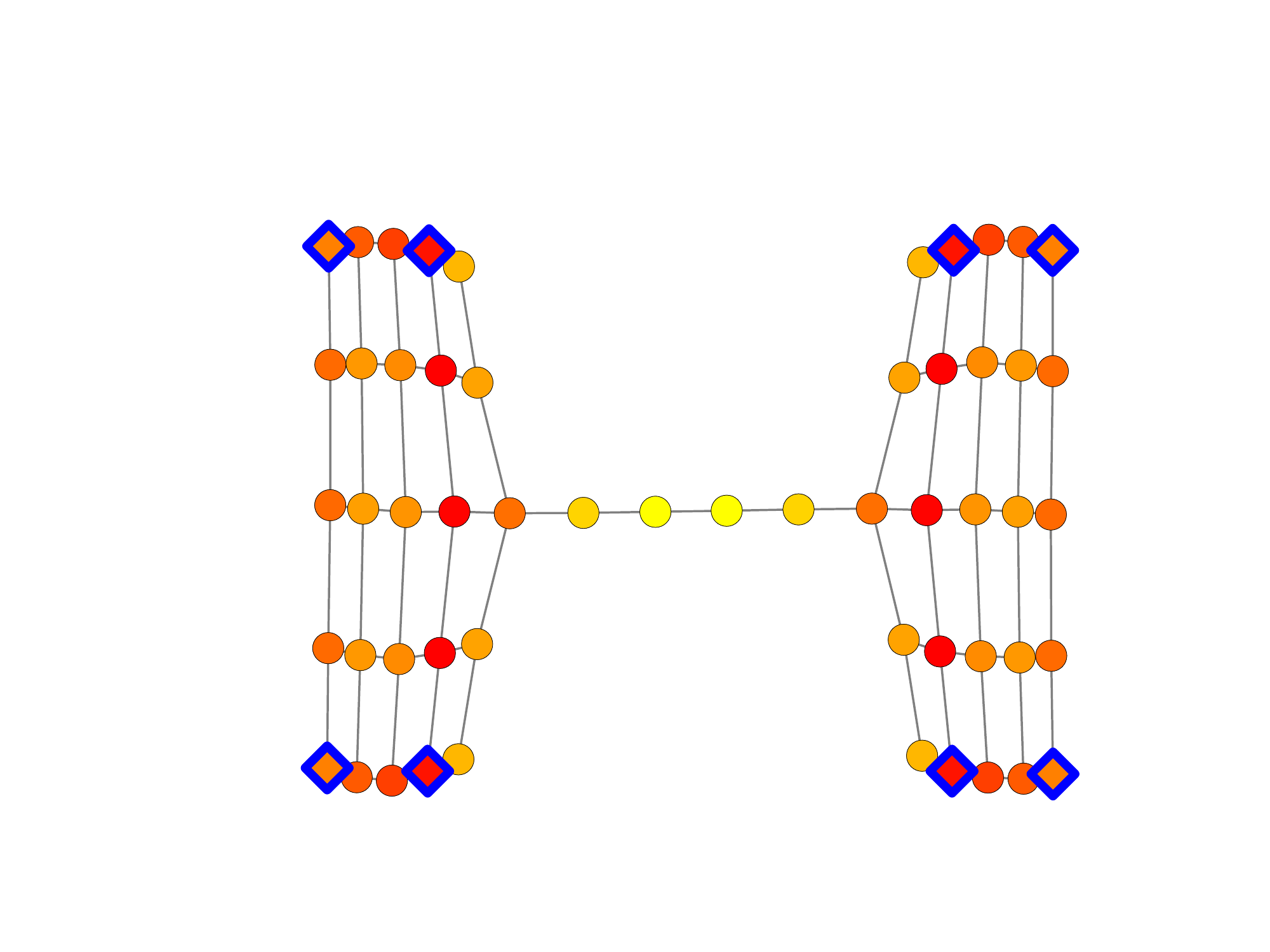}  \\
			(e) $|\Omega| = 3, \kappa = 4$ &(f)  $|\Omega| = 8, \kappa = 4$ \\ 
		\end{tabular}
	\end{center}
	\caption{Comparison of BC and SOC-BC. The blue diamond-shaped nodes represent nodes in $\Omega$. In (a), the nodes connecting the two components of the graph, which can be viewed as a "bridge" between two communities, have high centrality scores because they are essential for flow from one component to another. Whereas in (b) - (f), since the flow is limited to just 4 steps, the majoring of the flow in the graph would be within the two components thus the bridge nodes are no longer as important. }
	\label{compBC}
\end{figure}%

Differences between \soc-Katz and Katz centrality are illustrated in Figure \ref{compKatz}. In this experiment, we set $\kappa = 4$. The graph in Figure \ref{compKatz} (a) represents the standard Katz centrality. As expected, the nodes towards the center of the grid are the most important  nodes according to this model. The graphs in (b) - (f) show different scenarios where the nodes in $\Omega$ are marked with a blue-edge diamond-shaped node. As we can see in the Figure (e), depending on the value of $\kappa$, even with a relatively large ration of $|\Omega|/|V|$, (0.5 for (e)), the centrality scores can still be significantly different than the scores from their baseline models. 

A more comprehensive study comparing \soc-Katz with Katz on the grid graph is shown in Figure \ref{soc_grid_boxplot}. For each value of $\kappa$, with $2 \leq \kappa \leq 16$, we run 30 experiments. Each experiment consists of choosing nodes at random to be in the set $\Omega$. The boxplots in blue represent experiments with $|\Omega|/|V| = 0.1$, while $|\Omega|/|V| = 0.2$ are represented in red. We choose the values 0.1 and 0.2 for the ratio $|\Omega|/|V|$  because in most applications the set $\Omega$ will be considerably smaller than $V$. For example in social networks, the percentage of lurkers in an online community is estimated to range from 50 to 90 percent of the total membership \cite{katz1998luring,mason1999issues,soroka2003we}. As expected, the results show that as $\kappa$ increases, the correlation between \soc-Katz and Katz ranking increases. 
It is interesting to observe that since Katz centrality is based on infinite-length random walks emanating from a  node, it is not clear what value of $\kappa$ would make \soc-Katz identical to the standard Katz centrality for a given graph. However, this is not the case with betweenness centrality which is based on shortest paths. For a given graph, setting $\kappa$ to the longest shortest path would make \soc-BC identical to BC.

\begin{figure}[htp]
	\begin{center}
		\begin{tabular}{cc}
			
			\includegraphics[width=0.5\linewidth]{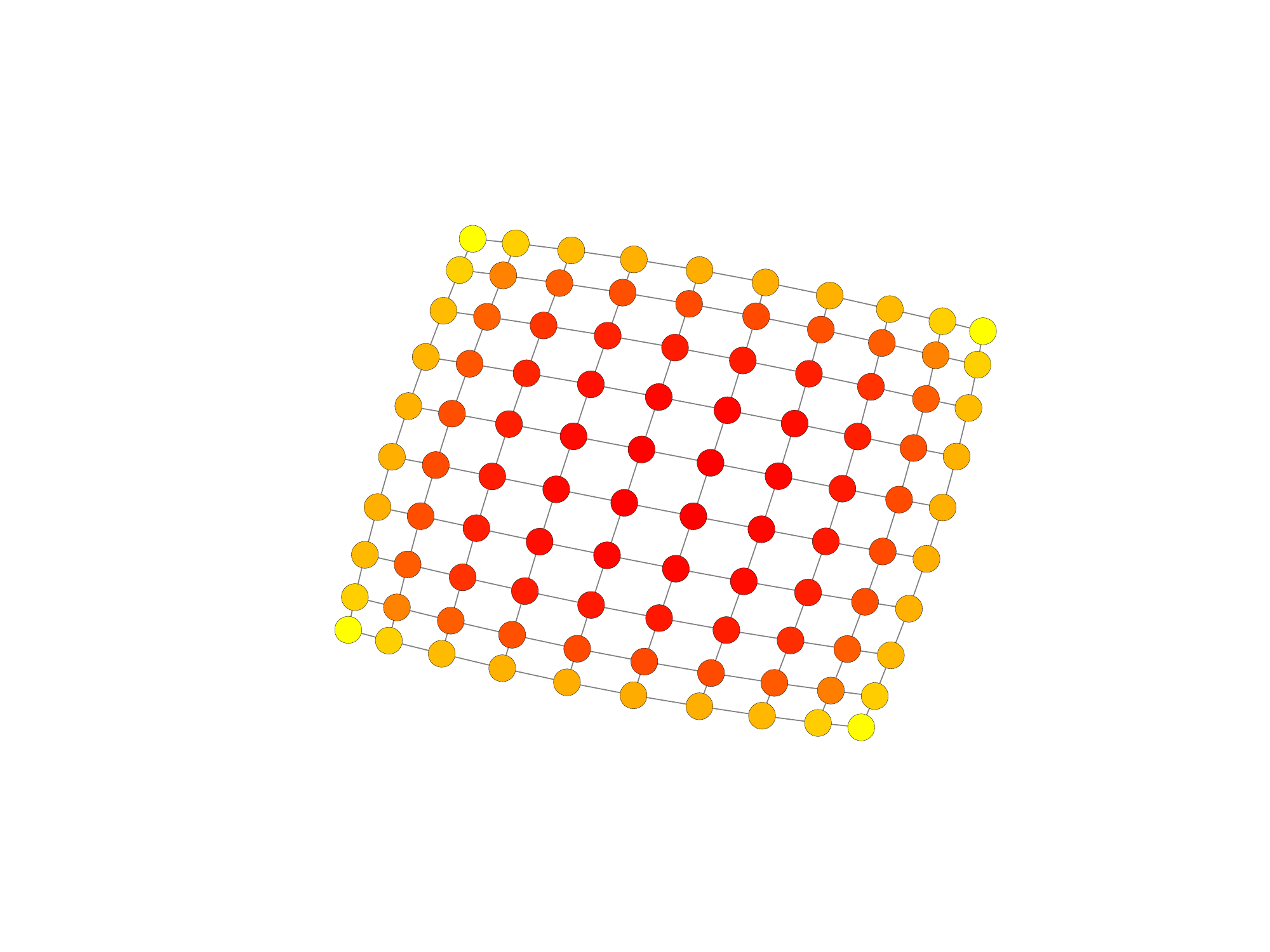} &     \includegraphics[width=0.5\linewidth]{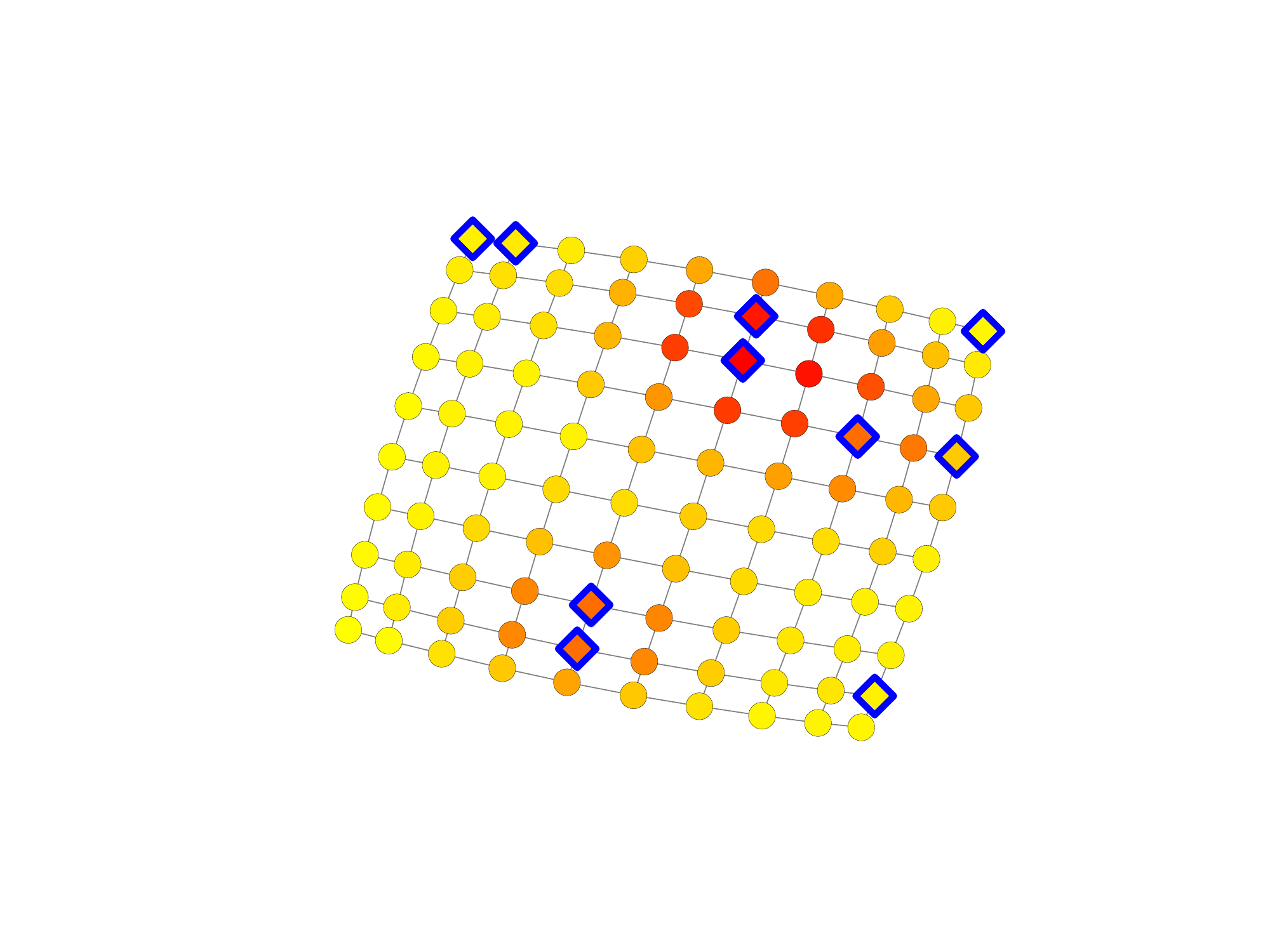} \\
			(a) Standard Katz  &(b) $|\Omega|/|V| = 0.1, \kappa = 4$\\ 
			\includegraphics[width=0.5\linewidth]{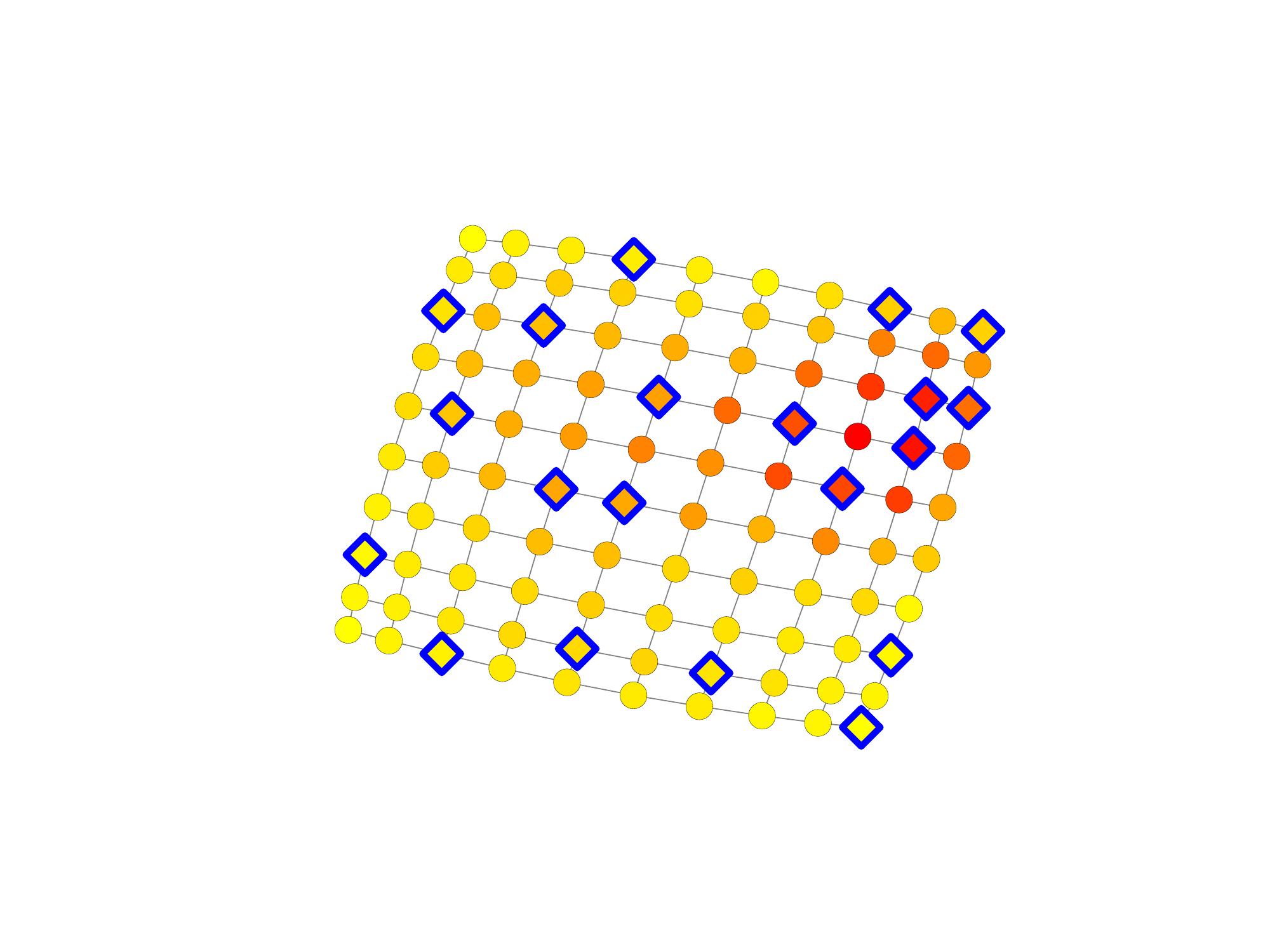}  &\includegraphics[width=0.5\linewidth]{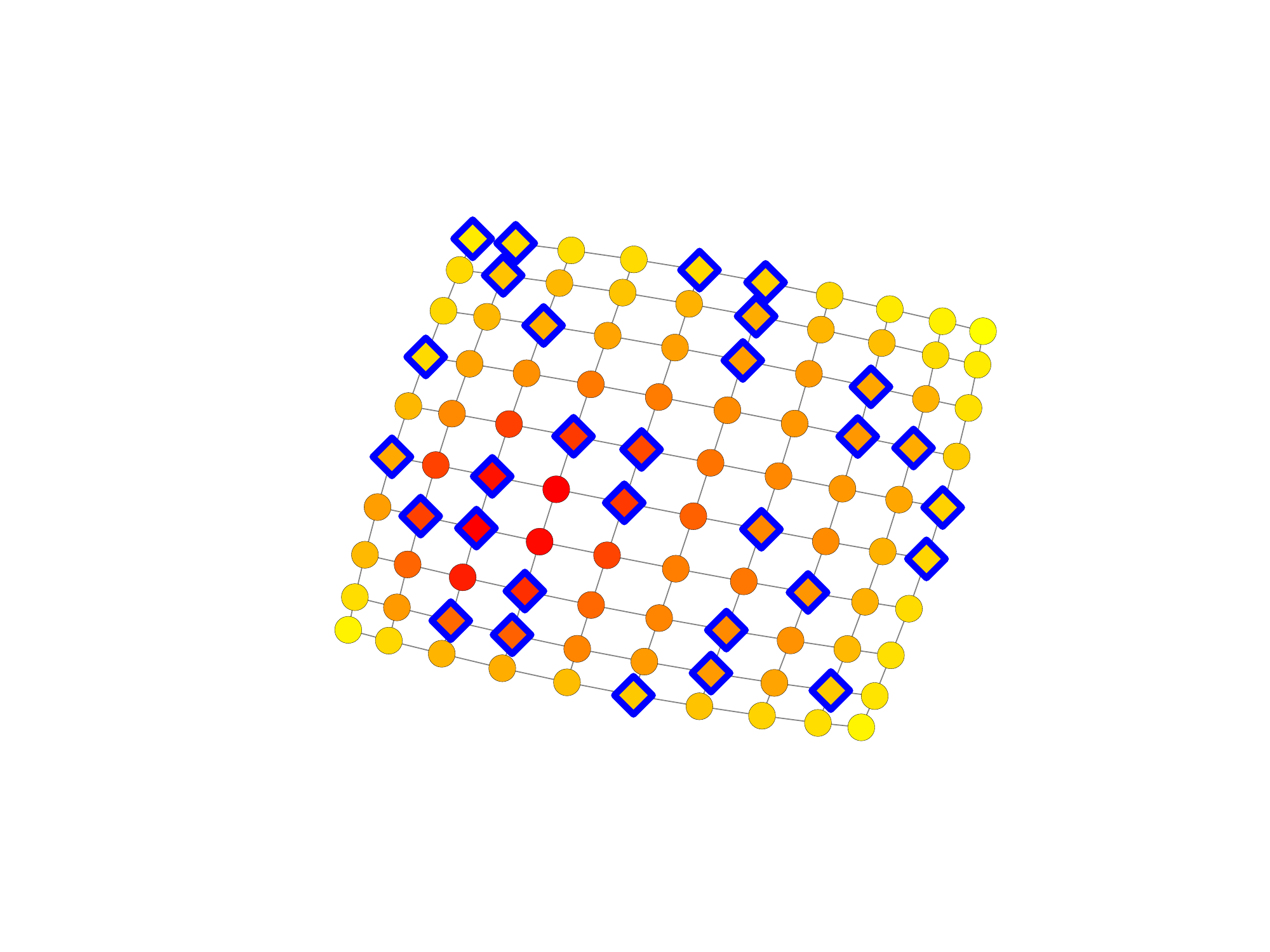}  \\
			(c) $|\Omega|/|V| = 0.2, \kappa = 4$ &(d) $|\Omega|/|V| = 0.3, \kappa = 4$\\ 
			\includegraphics[width=0.5\linewidth]{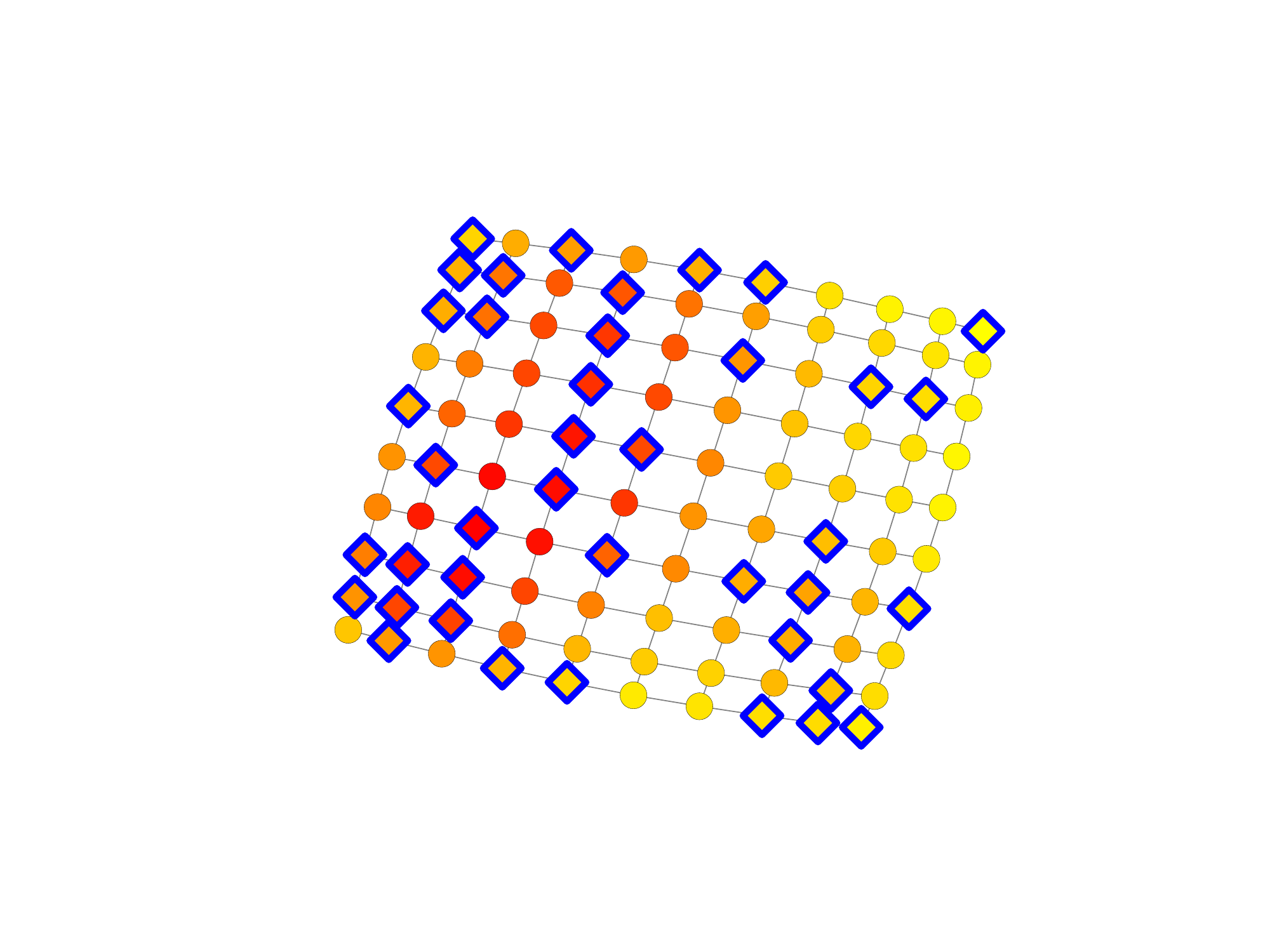}&\includegraphics[width=0.5\linewidth]{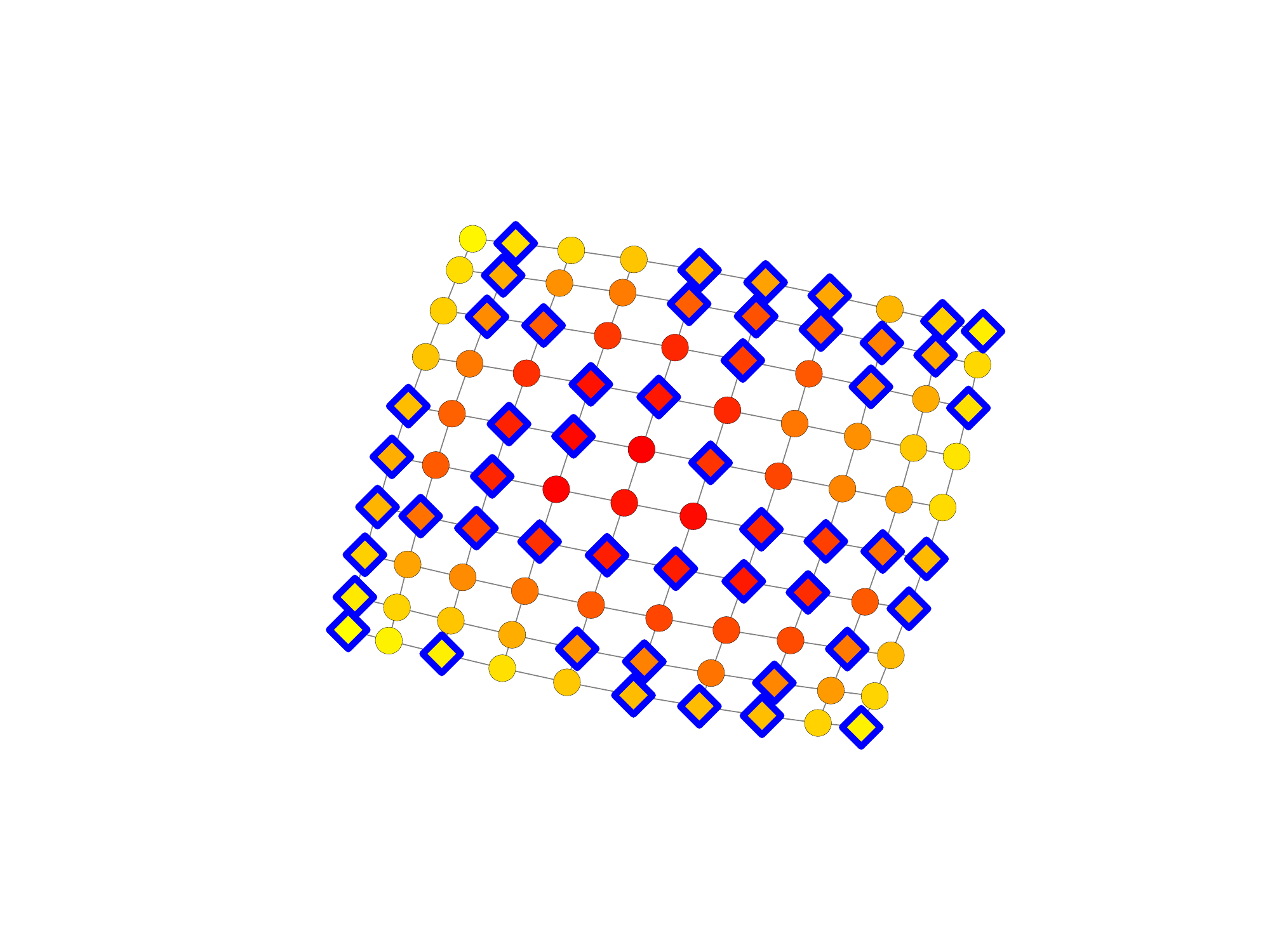}  \\
			(e) $|\Omega|/|V| = 0.4, \kappa = 4$ &(f) $|\Omega|/|V| = 0.5, \kappa = 4$ \\ 
		\end{tabular}
	\end{center}
	\caption{Comparison of Katz and SOC-Katz; The blue diamond-shaped nodes represent nodes in $\Omega$. In (a), the standard Katz centrality shows that the nodes more centrally located on the grid have a higher importance. For example, with respect to information spreading in a network, this implies that these nodes are the most influential in information spreading. However as observed from (b) to (f) if the information spread has fixed travel distance, then just the connectivity structure of the network is not enough to conclude about the most influential nodes. }
	\label{compKatz}
\end{figure}%

\begin{figure}[htb]
		\centering
	\begin{tikzpicture}
	\node (img)  {\includegraphics[width=0.5\textwidth]{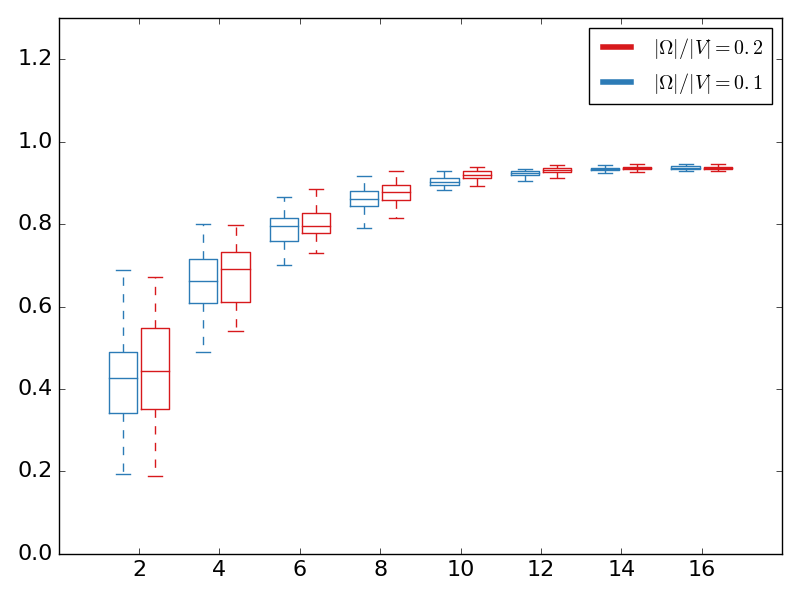}};
	\node[below=of img, node distance=0cm, yshift=1.3cm,font=\color{black}] {$\kappa$};
	\node[left=of img, node distance=0cm, rotate=90, xshift=0cm, anchor=center,yshift=-1.0cm,font=\color{black}] {Kendall $\tau$};
	\end{tikzpicture}
	\caption{SOC-Katz Vs. Katz for $10 \times 10$ Grid graph. }
	\label{soc_grid_boxplot}
\end{figure}

The second set of experiments in this section is carried out on real-world graphs, the Minnesota state road network, and the collaboration network.  In the Minnesota state road network, first, we compare \soc-BC with BC, while varying the value of $\kappa$, second, we compare $\soc$-RWBC and RWBC for a given source-target pair. The comparison of \soc-BC and BC is shown in Figure \ref{soc_MN_boxplot}. The parameter $\kappa$ is varied from 2 to 16. For each value of $\kappa$, we perform 30 experiments where each experiment consists of sampling a set of nodes $\Omega \subset V$ for a fixed ratio $|\Omega|/|V|$. We fix the ratio to 0.1 and 0.2, represented by blue and red boxplots respectively. Given that the Minnesota state road network has an average shortest path length of approximately 35.4, the results show that for values of $\kappa$, smaller than the average shortest path length, we can get significant differences between \soc-BC and BC ranking. Thus, we  find where the standard BC may potentially fail in identifying central nodes.  We visually demonstrate a similar result for \soc-RWBC and RWBC in Figure \ref{stRWBC}. In this experiment, we pick a pair of nodes representing a source and target and then compute the RWBC scores contributed by the two nodes referring to them as $s$-$t$-RWBC. Given that RWBC is identical to the current flow betweenness centrality \cite{brandes2005centrality}, one can think of this experiment as injecting a unit of current from the source flowing to target and measuring the fraction of current flowing through each node. The graph on Figure \ref{stRWBC} (a) represents the $s$-$t$-RWBC scores for the source-target pair represented by nodes in black. The $s$-$t$-RWBC values identical to zero are represented with nodes with negligible size. As  expected, the results show higher $s$-$t$-RWBC values for nodes close to the source and target. We perform a similar experiment for \soc-RWBC with $\kappa = 20$, while the distance from source to target is larger than 20. This implies that every random walk from source to target must pass through at least one node in $\Omega$. The results in Figure \ref{stRWBC} (b) show the $s$-$t$-\soc-RWBC values for the given source-target pair. In the case of \soc-RWBC, the higher central nodes are now the nodes close to the nodes in $\Omega$. This example demonstrates how \soc-RWBC can be used to identify congested nodes in a road network that is equipped with wireless charging lanes.

\begin{figure}[htb]
		\centering
	\begin{tikzpicture}
	\node (img)  {\includegraphics[width=0.5\textwidth]{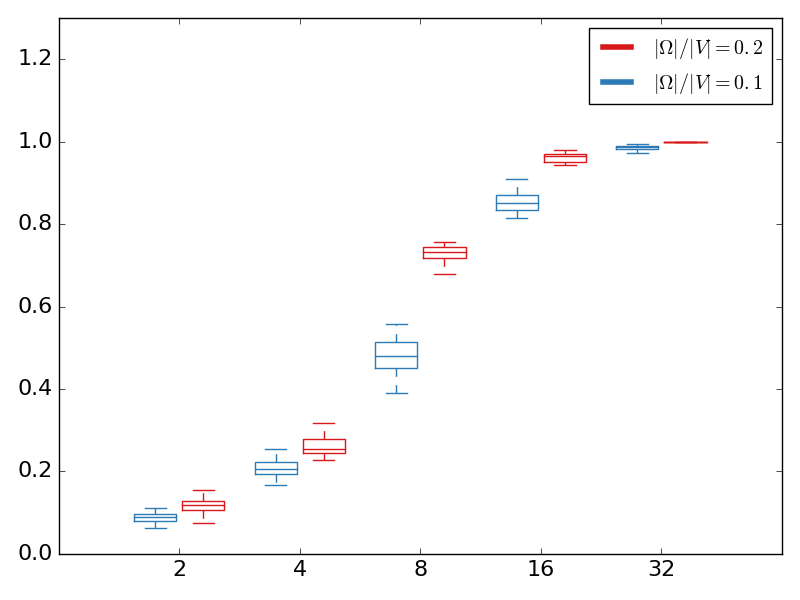}};
	\node[below=of img, node distance=0cm, yshift=1.3cm,font=\color{black}] {$\kappa$ };
	\node[left=of img, node distance=0cm, xshift=0.5cm, rotate=90, anchor=center,yshift=-0.5cm,font=\color{black}] {Kendall $\tau$};
	\end{tikzpicture}
	\caption{SOC-BC Vs BC for Minnesota State Road Network }
	\label{soc_MN_boxplot}
\end{figure}

\begin{figure}[htp]
	\begin{center}
		\begin{tabular}{c}
			
			\includegraphics[width=0.5\textwidth]{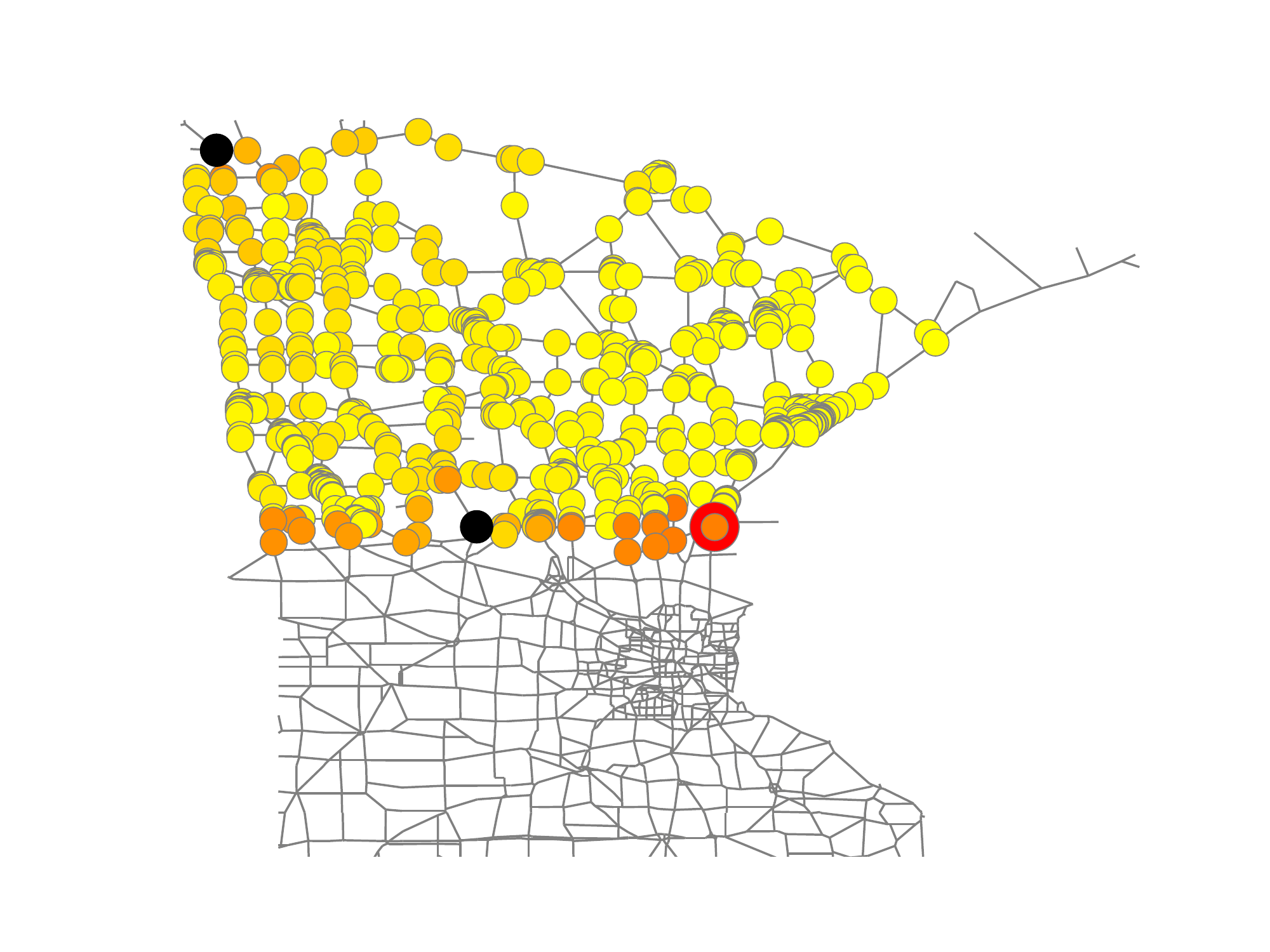}
			\includegraphics[width=0.5\textwidth]{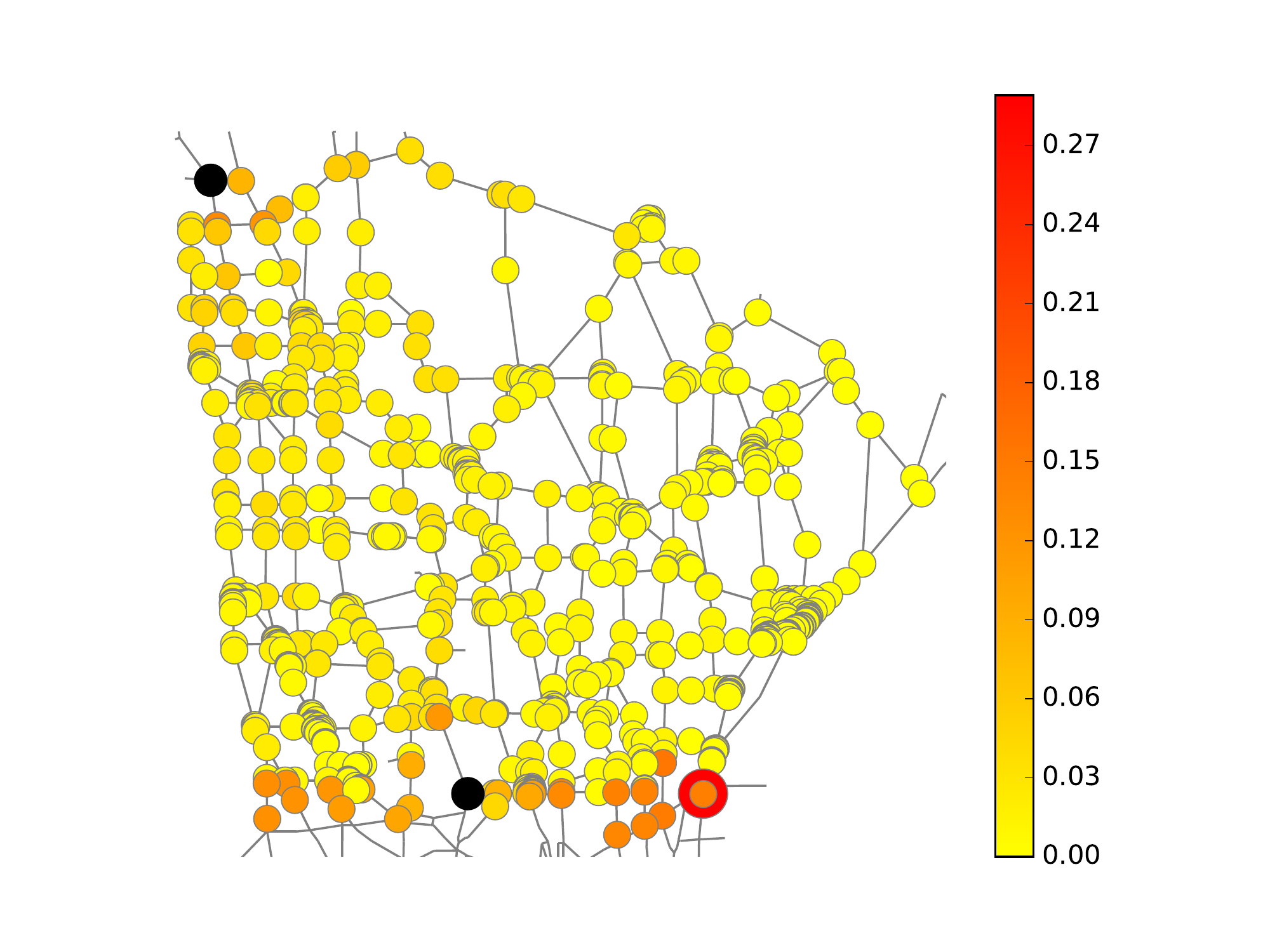}\\
			(a) \\
			\includegraphics[width=0.5\textwidth]{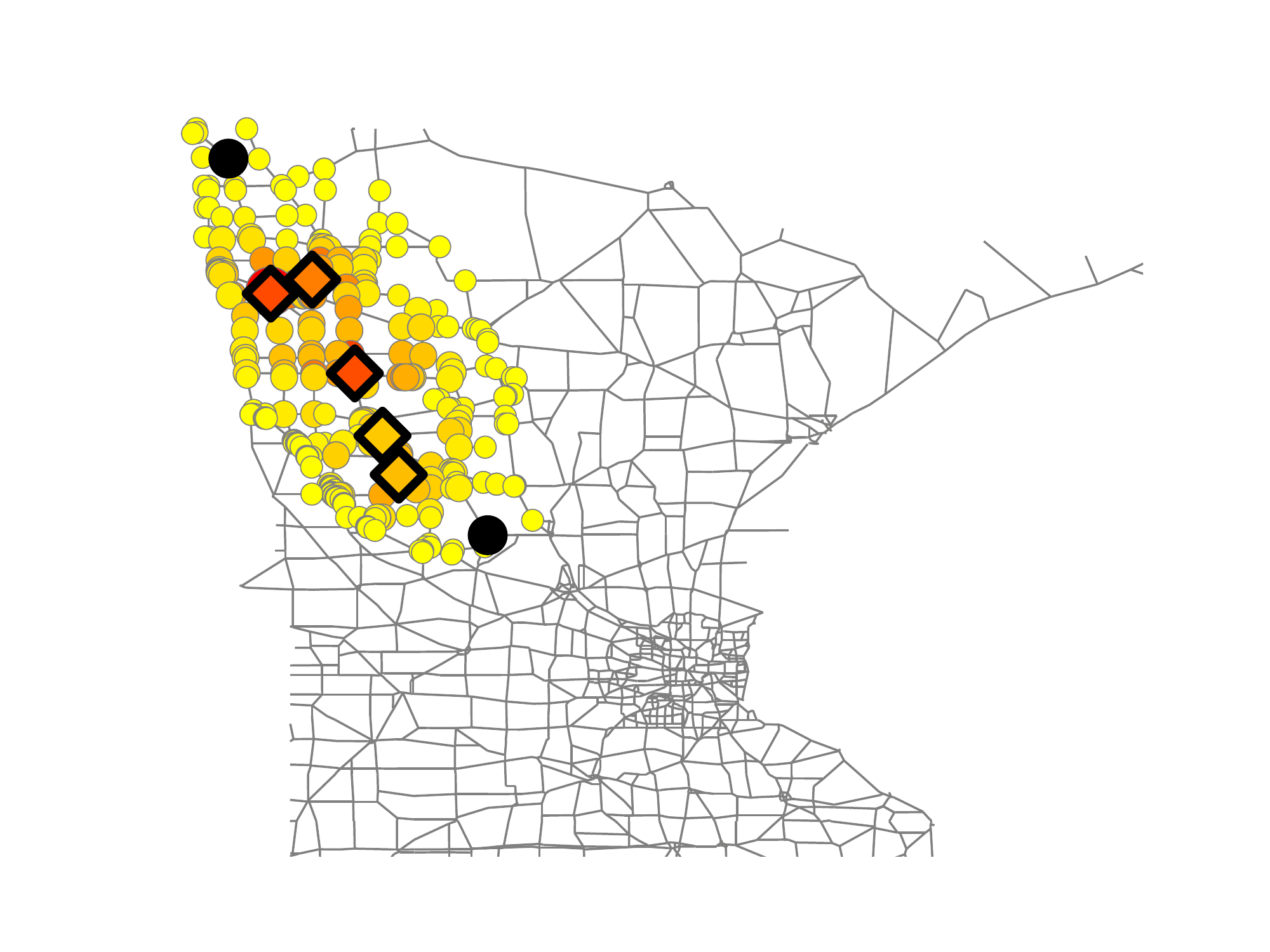}
			\includegraphics[width=0.5\textwidth]{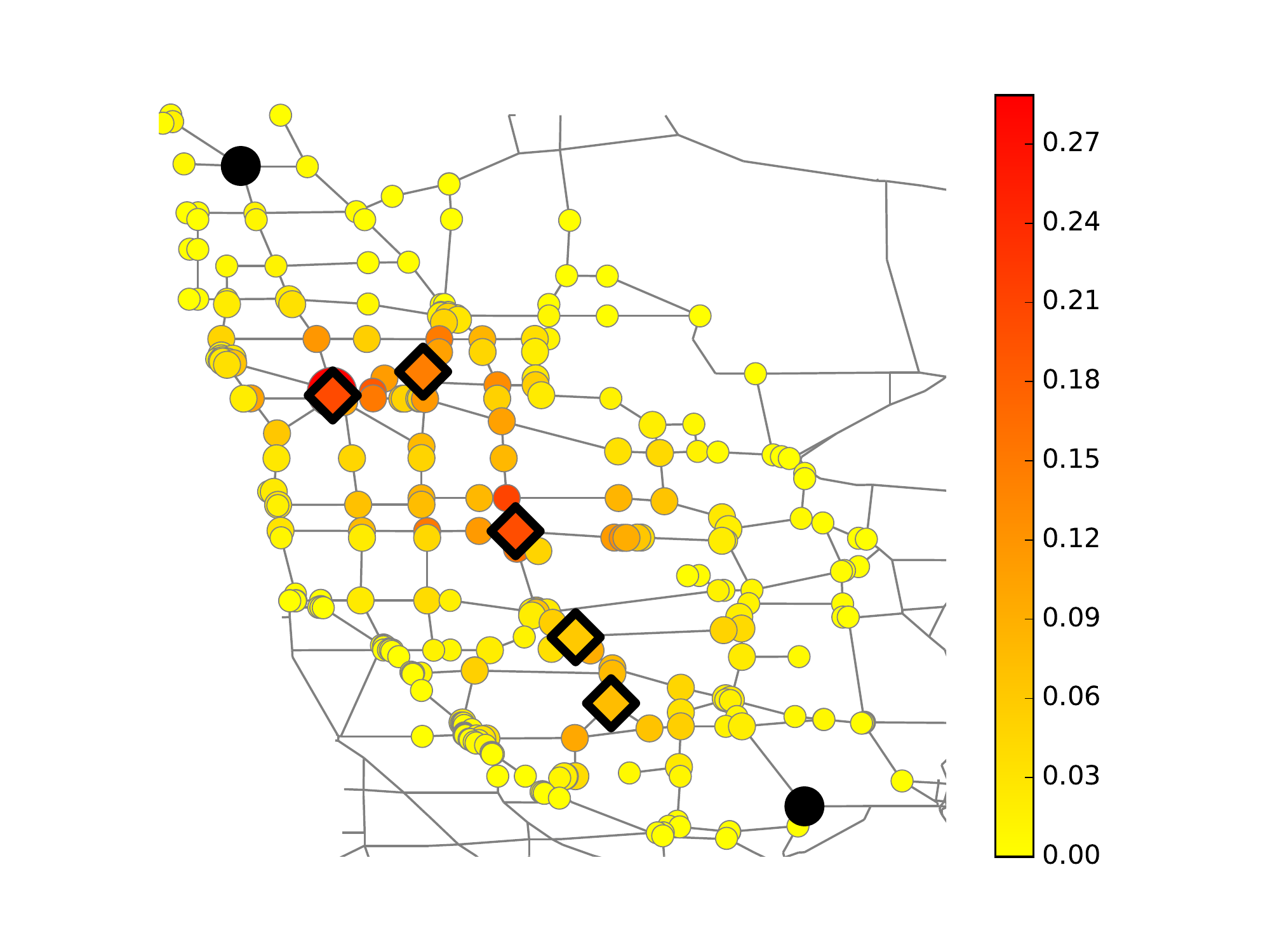}\\
			(b)
		\end{tabular}
	\end{center}
	\caption{Comparison of RWBC with SOC-RWBC for a single $s$-$t$ pair over the Minnesota state road network. The top-left and center black nodes represent the source and target nodes respectively. Nodes with $s$-$t$-centrality scores equal to 0 have have negligible node sizes.  (a) RWBC for a given $s$-$t$ pair. (b) SOC-RWBC for a given $s$-$t$ pair. Nodes with a triangular marker represent nodes in $\Omega$.}
	\label{stRWBC}
\end{figure}%
A comparison of \soc-Katz versus standard Katz with variations of the parameter $\kappa$ on the Collaboration network is shown in Figure  \ref{soc_rand_boxplot}. In general, for different values of $\kappa$, the correlation of \soc-Katz and Katz is high, generally above 0.8. However, once we plot the different ranking we see significant differences with the node rankings of the two measures. Thus, Kendall Tau correlation does not give a complete picture for this network. In particular, a node that is ranked highly with the standard Katz centrality can have a significantly less rank in a ranking with \soc-Katz. However, conversely, highly ranked nodes with \soc-Katz generally tend to also be highly ranked with respect to the standard Katz. With respect to the application to posters and lurkers in social networks, this follows the intuition that a user with a large number of neighbors (friends) can still be non-influential if the user together with all his/her neighbors (friends) are lurkers. On the other hand, a highly influential node would generally have many neighbors (friends). 

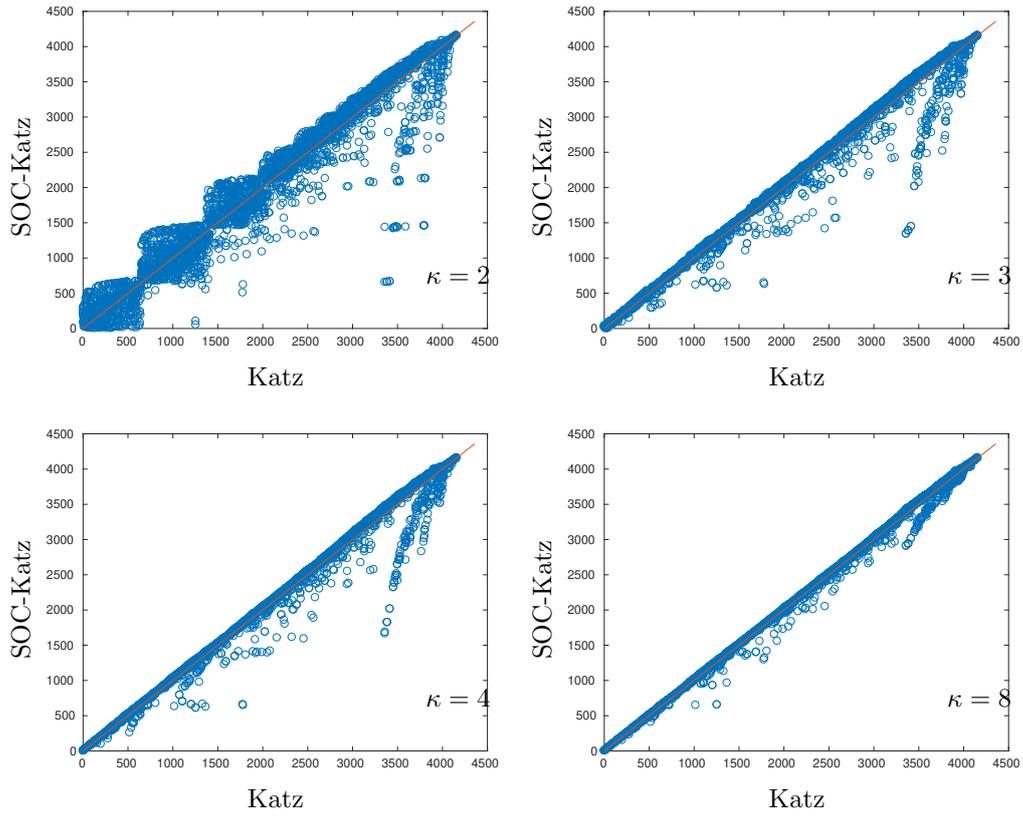
\begin{figure}[htb]
		\centering
	\begin{minipage}{0.45\textwidth}
		\begin{tikzpicture}
		\node (img)  {\includegraphics[width=1\textwidth]{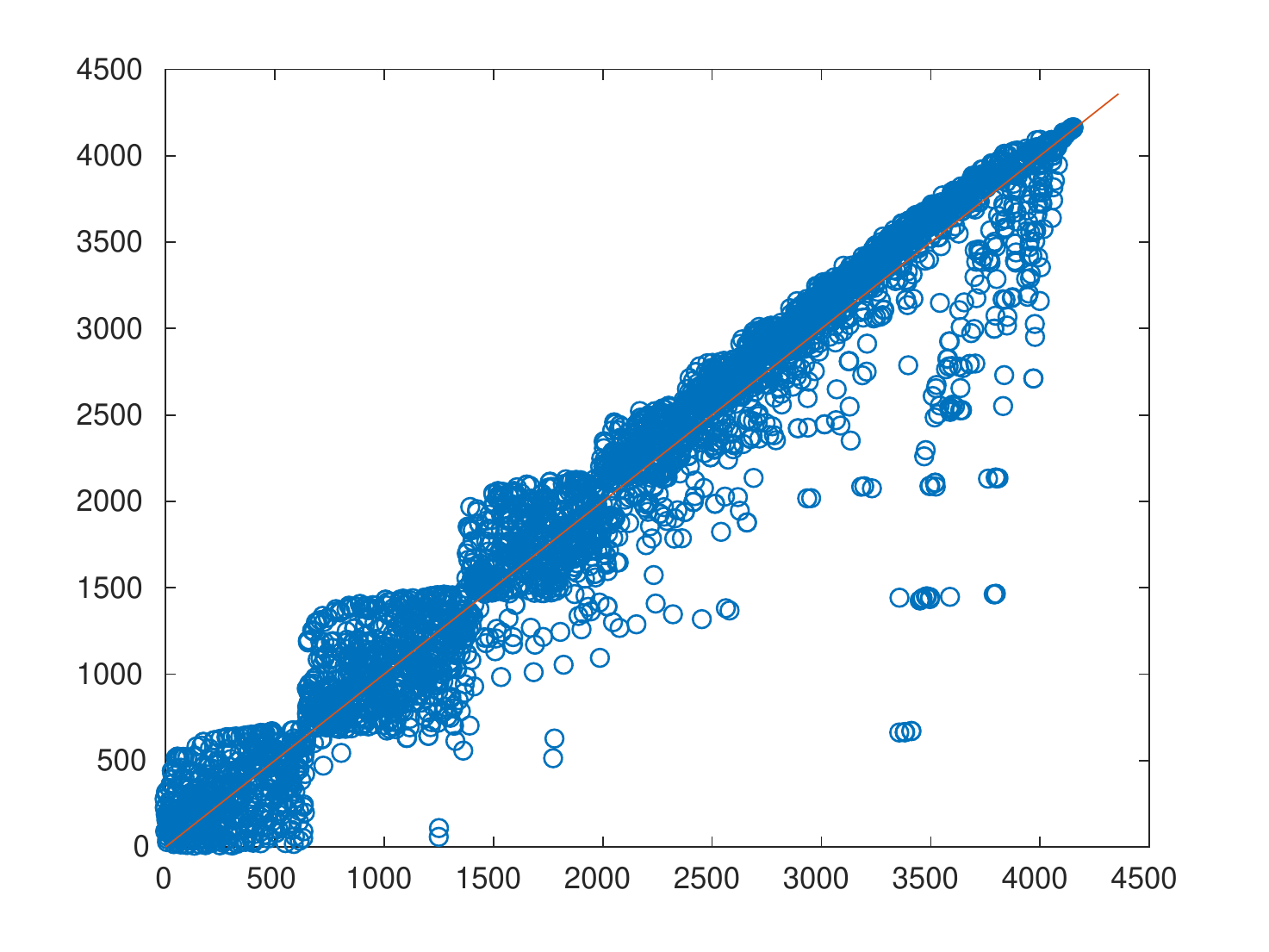}};
		\node[below=of img, node distance=0cm, yshift=1.3cm,font=\color{black}] {Katz};
		\node[left=of img, node distance=0cm, xshift=0.5cm, rotate=90, anchor=center,yshift=-0.7cm,font=\color{black}] {SOC-Katz};
		\node[right=of img, node distance=0cm, yshift=-1.3cm, xshift=-2.7cm, font=\color{black}]  {$\kappa=2$};
		\end{tikzpicture}
	\end{minipage}%
	\begin{minipage}{0.45\textwidth}
		\begin{tikzpicture}
		\node (img)  {\includegraphics[width=1\textwidth]{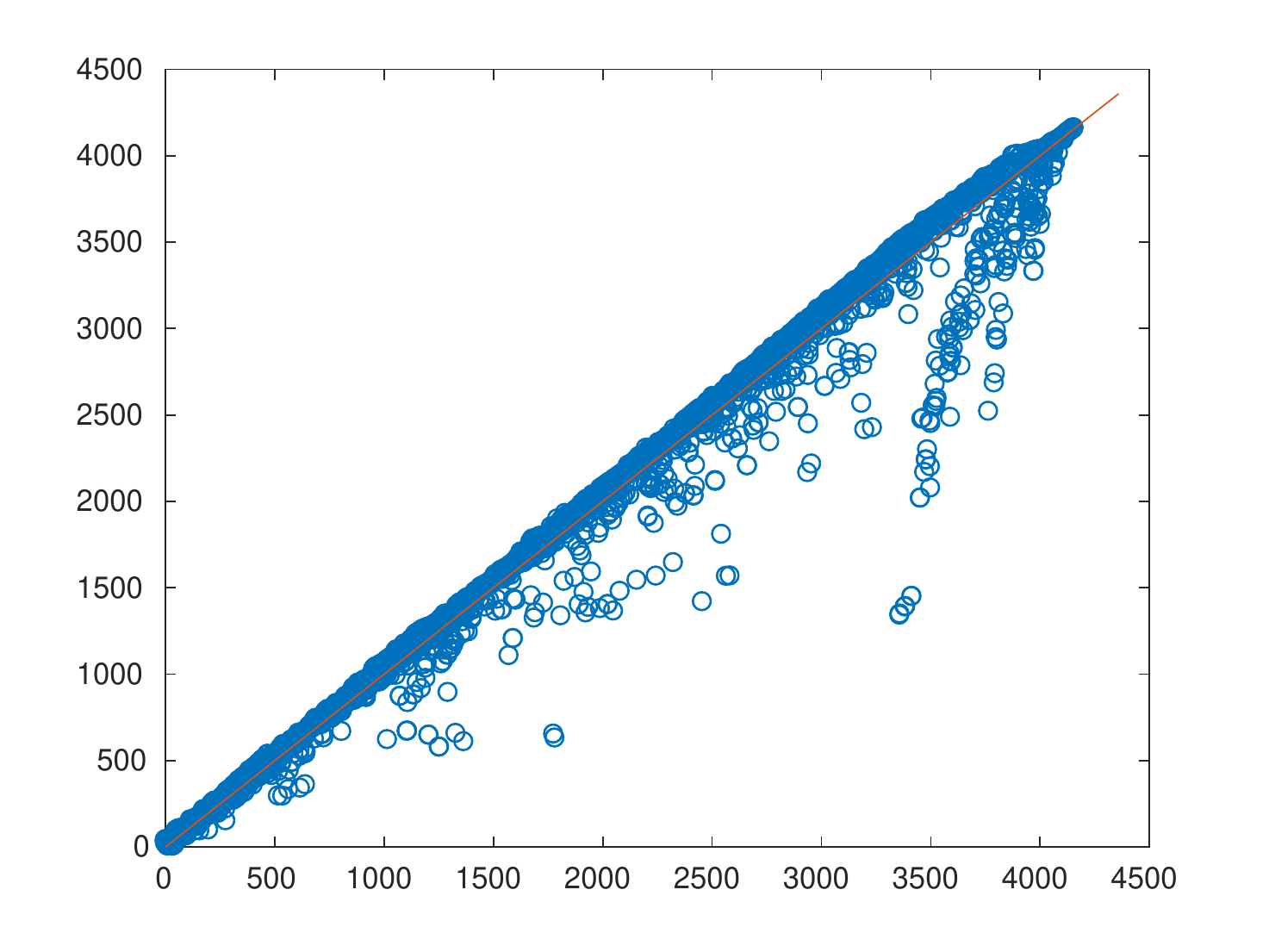}};
		\node[below=of img, node distance=0cm, yshift=1.3cm,font=\color{black}]  {Katz};
		\node[left=of img, node distance=0cm, rotate=90, xshift=0cm, anchor=center,yshift=-1.2cm,font=\color{black}] {SOC-Katz};
		\node[right=of img, node distance=0cm, yshift=-1.3cm, xshift=-2.7cm, font=\color{black}]  {$\kappa=3$};
		\end{tikzpicture}
	\end{minipage}%
	
	\begin{minipage}{0.45\textwidth}
		\begin{tikzpicture}
		\node (img)  {\includegraphics[width=1\textwidth]{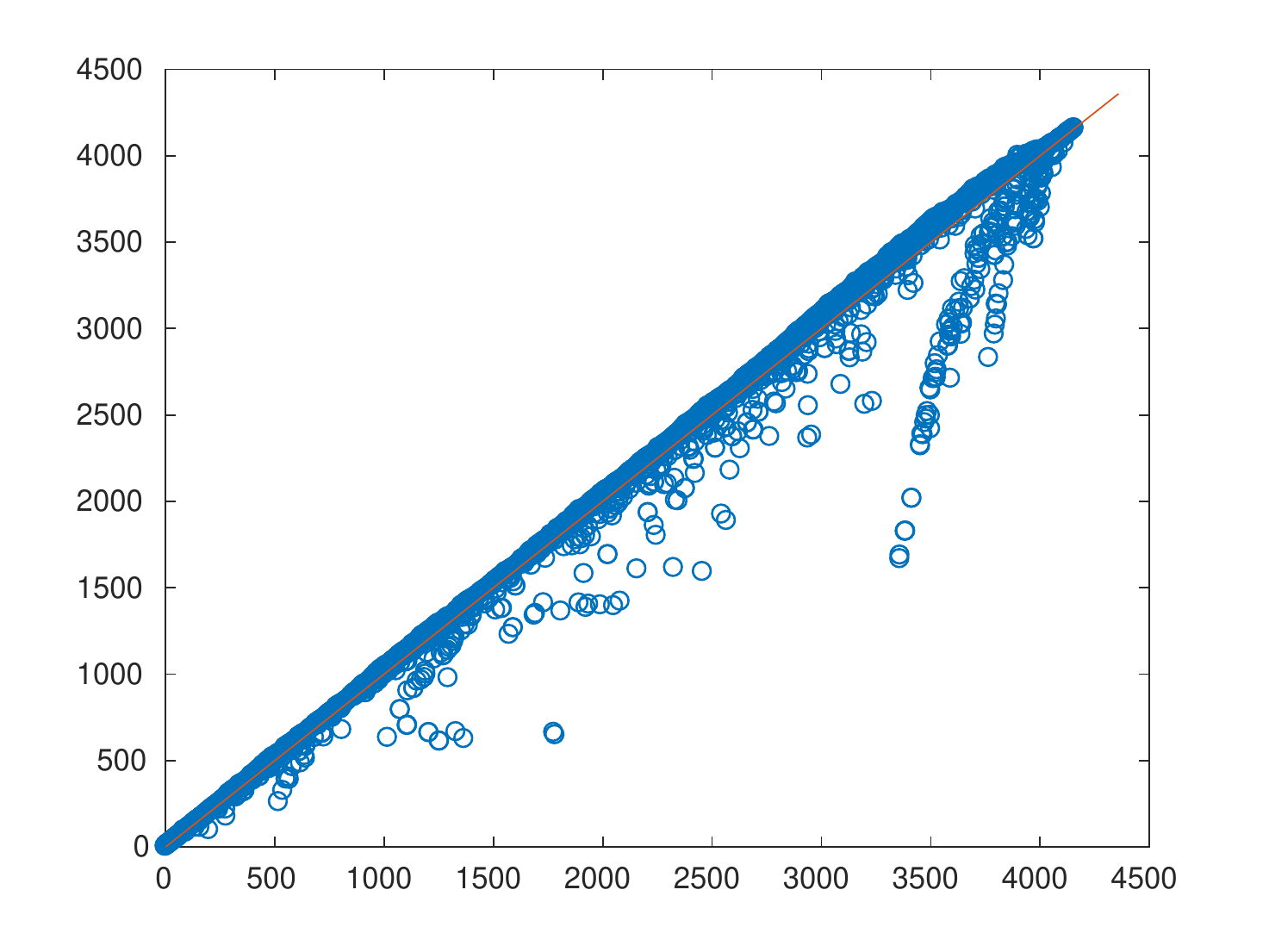}};
		\node[below=of img, node distance=0cm, yshift=1.3cm,font=\color{black}]  {Katz};
		\node[left=of img, node distance=0cm, rotate=90, xshift=0cm,anchor=center,yshift=-1.2cm,font=\color{black}] {SOC-Katz};
		\node[right=of img, node distance=0cm, yshift=-1.3cm, xshift=-2.7cm, font=\color{black}]  {$\kappa=4$};
		\end{tikzpicture}
	\end{minipage}%
	\begin{minipage}{0.45\textwidth}
		\begin{tikzpicture}
		\node (img)  {\includegraphics[width=1\textwidth]{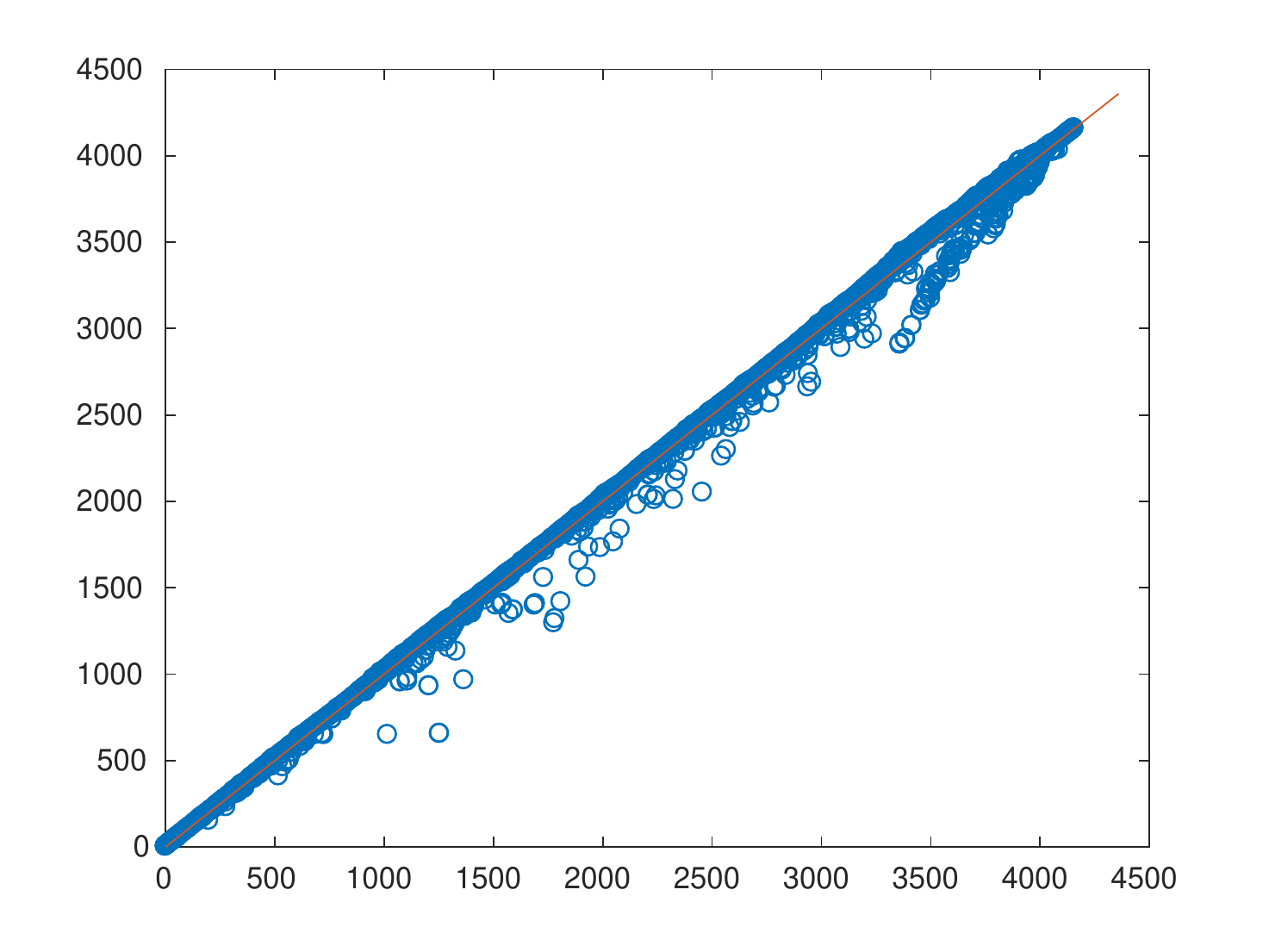}};
		\node[below=of img, node distance=0cm, yshift=1.3cm,font=\color{black}]  {Katz};
		\node[left=of img, node distance=0cm, rotate=90, xshift=0cm, anchor=center,yshift=-1.2cm,font=\color{black}] {SOC-Katz};
		\node[right=of img, node distance=0cm, yshift=-1.3cm, xshift=-2.7cm, font=\color{black}] {$\kappa=8$};
		\end{tikzpicture}
	\end{minipage}%

	\caption{Comparison of node rankings based on SOC-Katz and Katz for different values of $\kappa$ with $|\Omega|/|V| = 0.1$. Even if the above ranking give high Kendall Tau correlation, we notice that with the introduction of a ranking based SOC-Katz, highly ranked Katz nodes can significantly loose their ranking, however, less important nodes do not significantly increase with rank }
	\label{soc_rand_boxplot}
\end{figure}

\section{Conclusion}
An estimation of node spreading influence in a network is an important step towards understanding and controlling the spreading dynamics over the network. Centrality measures are traditionally used to identify influential nodes in a network. In this work, we extend the well-known measures of Katz, betweenness, and random-walk betweenness centralities to models that accommodate a resource, necessary for the spread, being consumed along the way. We present algorithms to compute the proposed centrality measures and carry out experiments on real-world networks. Lastly, we demonstrate simulation models that describe the flow process and show that they are highly correlated to the proposed centralities.  In order to answer the question, ``How good are these centrality measures?'', we analyze the centrality measures from three different perspectives, namely, usability, robustness, and novelty. From the usability perspective, among other experiments, we demonstrate how the centrality measures can be used to identify congested nodes in computer and road networks. From the robustness and novelty perspective, in the application of posters and lurkers in a social network, we showed that the proposed extension of Katz centrality follows the intuition that a user with a large number of neighbors (friends) can still be non-influential if the user together with all his/her neighbors (connections) are lurkers. On the other hand, a highly influential node would generally have many neighbors (connections).

The proposed measures take into account a spreading process that depends on a resource, such that the spreading process would be impossible without. Our numerical experiments demonstrate that the proposed measures differ significantly from the original measures when the resource is limited. On the other hand, they become identical to the original measures as the quantity of resource available tends to infinity. As a result, the proposed measures give a new tool and perspective to different application domains. For example, in the application of a road network equipped with wireless charging lanes, an optimal placement of these lanes with respect to traffic distribution, could be one where the distribution of centrality scores of all nodes is taken into account. In another domain, high centrality nodes can be considered for targeted attack or immunization strategies.

For a given application, the choice of which centrality measure to use to draw a conclusion about the network is extremely important as using a wrong measure can lead to meaningless results. The measures of \soc-RWBC and \soc-Katz are both based on random walks on the network. It is, however, important to note that, just as the standard measures, the random walks associated with \soc-RWBC have a fixed source and target node, while the random walks associated with \soc-Katz only have a fixed source node. Thus \soc-Katz is more suitable for applications where the flow process does not have a specified destination, for example, a disease spread. The \soc-RWBC and subsequently \soc-BC are suitable for applications where the flow process has a specified destination, for example vehicle flow.

There are numerous future research directions associated with the resource consumption based centralities. For example, we propose to  explore other fundamental network properties such as connectedness, clustering, and network robustness  in the context of consumable resource networks. In this work, there exists a set of nodes that facilitate flow in the network, conversely, problems in network interdiction \cite{wood1993deterministic} deal with the identification of nodes that hinder flow. An interesting direction is to explore the relationship between these two problems in more detail. 
Another highly relevant direction for the future work is to consider a distribution of resource consumption based  centralities in realistic network generation \cite{gutfraind2015multiscale,staudt2017generating} . This is particularly important for the simulation and verification studies. To the best of our knowledge, no generating model currently considers a distribution of resource consumption based centralities. Also, the resource consumption models can be generalized for clusters and communities. Moreover, one of the natural extensions of this work is introducing resource consumption element to node and edge similarity measures.

\section*{Acknowledgements}

We thank the anonymous reviewers whose comments and suggestions helped to improve and clarify this manuscript. This research is supported by the National Science Foundation under Award \#1647361. Any opinions, findings, conclusions or recommendations expressed in this material are those of the authors and do not necessarily reflect the views of the National Science Foundation.


\begin{thebibliography}{10}
	
	\bibitem{albert2002statistical}
	R{\'e}ka Albert and Albert-L{\'a}szl{\'o} Barab{\'a}si.
	\newblock Statistical mechanics of complex networks.
	\newblock {\em Reviews of modern physics}, 74(1):47, 2002.
	
	\bibitem{albert2000error}
	R{\'e}ka Albert, Hawoong Jeong, and Albert-L{\'a}szl{\'o} Barab{\'a}si.
	\newblock Error and attack tolerance of complex networks.
	\newblock {\em nature}, 406(6794):378, 2000.
	
	\bibitem{altshuler2011augmented}
	Yaniv Altshuler, Rami Puzis, Yuval Elovici, Shlomo Bekhor, and AS~Pentland.
	\newblock Augmented betweenness centrality for mobility prediction in
	transportation networks.
	\newblock In {\em International Workshop on Finding Patterns of Human Behaviors
		in NEtworks and MObility Data, NEMO11}, 2011.
	
	\bibitem{bae2014identifying}
	Joonhyun Bae and Sangwook Kim.
	\newblock Identifying and ranking influential spreaders in complex networks by
	neighborhood coreness.
	\newblock {\em Physica A: Statistical Mechanics and its Applications},
	395:549--559, 2014.
	
	\bibitem{barrat2008dynamical}
	Alain Barrat, Marc Barthelemy, and Alessandro Vespignani.
	\newblock {\em Dynamical processes on complex networks}.
	\newblock Cambridge university press, 2008.
	
	\bibitem{basu2013state}
	Anirban Basu, Simon Fleming, James Stanier, Stephen Naicken, Ian Wakeman, and
	Vijay~K Gurbani.
	\newblock The state of peer-to-peer network simulators.
	\newblock {\em {ACM Computing Surveys (CSUR)}}, 45(4):46, 2013.
	
	\bibitem{bavelas1948mathematical}
	Alex Bavelas.
	\newblock A mathematical model for group structures.
	\newblock {\em Human organization}, 7(3):16--30, 1948.
	
	\bibitem{bettencourt2006power}
	Lu{\'\i}s~MA Bettencourt, Ariel Cintr{\'o}n-Arias, David~I Kaiser, and Carlos
	Castillo-Ch{\'a}vez.
	\newblock The power of a good idea: Quantitative modeling of the spread of
	ideas from epidemiological models.
	\newblock {\em Physica A: Statistical Mechanics and its Applications},
	364:513--536, 2006.
	
	\bibitem{bi2016review}
	Zicheng Bi, Tianze Kan, Chunting~Chris Mi, Yiming Zhang, Zhengming Zhao, and
	Gregory~A Keoleian.
	\newblock A review of wireless power transfer for electric vehicles: Prospects
	to enhance sustainable mobility.
	\newblock {\em Applied Energy}, 179:413--425, 2016.
	
	\bibitem{boccaletti2006complex}
	Stefano Boccaletti, Vito Latora, Yamir Moreno, Martin Chavez, and D-U Hwang.
	\newblock Complex networks: Structure and dynamics.
	\newblock {\em Physics reports}, 424(4-5):175--308, 2006.
	
	\bibitem{bonacich1987power}
	Phillip Bonacich.
	\newblock Power and centrality: A family of measures.
	\newblock {\em American journal of sociology}, 92(5):1170--1182, 1987.
	
	\bibitem{bonacich1991simultaneous}
	Phillip Bonacich.
	\newblock Simultaneous group and individual centralities.
	\newblock {\em Social networks}, 13(2):155--168, 1991.
	
	\bibitem{borgatti2005centrality}
	Stephen~P Borgatti.
	\newblock Centrality and network flow.
	\newblock {\em Social networks}, 27(1):55--71, 2005.
	
	\bibitem{brandes2001faster}
	Ulrik Brandes.
	\newblock A faster algorithm for betweenness centrality.
	\newblock {\em Journal of mathematical sociology}, 25(2):163--177, 2001.
	
	\bibitem{brandes2005centrality}
	Ulrik Brandes and Daniel Fleischer.
	\newblock Centrality measures based on current flow.
	\newblock In {\em STACS}, volume 3404, pages 533--544. Springer, 2005.
	
	\bibitem{chen2011algebraic}
	Jie Chen and Ilya Safro.
	\newblock Algebraic distance on graphs.
	\newblock {\em {SIAM Journal on Scientific Computing}}, 33(6):3468--3490, 2011.
	
	\bibitem{chen2016optimal}
	Zhibin Chen, Fang He, and Yafeng Yin.
	\newblock Optimal deployment of charging lanes for electric vehicles in
	transportation networks.
	\newblock {\em Transportation Research Part B: Methodological}, 91:344--365,
	2016.
	
	\bibitem{cirimele2014wireless}
	Vincenzo Cirimele, Fabio Freschi, and Paolo Guglielmi.
	\newblock Wireless power transfer structure design for electric vehicle in
	charge while driving.
	\newblock In {\em {Electrical Machines (ICEM), 2014 International Conference
			on}}, pages 2461--2467. IEEE, 2014.
	
	\bibitem{cohen2001breakdown}
	Reuven Cohen, Keren Erez, Daniel Ben-Avraham, and Shlomo Havlin.
	\newblock Breakdown of the internet under intentional attack.
	\newblock {\em Physical review letters}, 86(16):3682, 2001.
	
	\bibitem{colizza2007reaction}
	Vittoria Colizza, Romualdo Pastor-Satorras, and Alessandro Vespignani.
	\newblock Reaction--diffusion processes and metapopulation models in
	heterogeneous networks.
	\newblock {\em Nature Physics}, 3(4):276, 2007.
	
	\bibitem{crucitti2006centrality}
	Paolo Crucitti, Vito Latora, and Sergio Porta.
	\newblock Centrality in networks of urban streets.
	\newblock {\em Chaos: an interdisciplinary journal of nonlinear science},
	16(1):015113, 2006.
	
	\bibitem{davis2011university}
	Timothy~A Davis and Yifan Hu.
	\newblock The university of florida sparse matrix collection.
	\newblock {\em ACM Transactions on Mathematical Software (TOMS)}, 38(1):1,
	2011.
	
	\bibitem{diekmann2000mathematical}
	Odo Diekmann and Johan Andre~Peter Heesterbeek.
	\newblock {\em Mathematical epidemiology of infectious diseases: model
		building, analysis and interpretation}, volume~5.
	\newblock John Wiley \& Sons, 2000.
	
	\bibitem{fouss2007random}
	Francois Fouss, Alain Pirotte, Jean-Michel Renders, and Marco Saerens.
	\newblock Random-walk computation of similarities between nodes of a graph with
	application to collaborative recommendation.
	\newblock {\em {IEEE Transactions on knowledge and data engineering}},
	19(3):355--369, 2007.
	
	\bibitem{freeman1978centrality}
	Linton~C Freeman.
	\newblock Centrality in social networks conceptual clarification.
	\newblock {\em Social networks}, 1(3):215--239, 1978.
	
	\bibitem{freeman1991centrality}
	Linton~C Freeman, Stephen~P Borgatti, and Douglas~R White.
	\newblock Centrality in valued graphs: A measure of betweenness based on
	network flow.
	\newblock {\em Social networks}, 13(2):141--154, 1991.
	
	\bibitem{fuller2016wireless}
	Micah Fuller.
	\newblock Wireless charging in california: Range, recharge, and vehicle
	electrification.
	\newblock {\em {Transportation Research Part C: Emerging Technologies}},
	67:343--356, 2016.
	
	\bibitem{ghosh2012rethinking}
	Rumi Ghosh and Kristina Lerman.
	\newblock Rethinking centrality: the role of dynamical processes in social
	network analysis.
	\newblock {\em arXiv preprint arXiv:1209.4616}, 2012.
	
	\bibitem{guimera2002optimal}
	Roger Guimer{\`a}, Albert D{\'\i}az-Guilera, Fernando Vega-Redondo, Antonio
	Cabrales, and Alex Arenas.
	\newblock Optimal network topologies for local search with congestion.
	\newblock {\em Physical review letters}, 89(24):248701, 2002.
	
	\bibitem{gutfraind2015multiscale}
	Alexander Gutfraind, Ilya Safro, and Lauren~Ancel Meyers.
	\newblock Multiscale network generation.
	\newblock In {\em Information Fusion (Fusion), 2015 18th International
		Conference on}, pages 158--165. IEEE, 2015.
	
	\bibitem{hethcote2000mathematics}
	Herbert~W Hethcote.
	\newblock The mathematics of infectious diseases.
	\newblock {\em SIAM review}, 42(4):599--653, 2000.
	
	\bibitem{holme2003congestion}
	Petter Holme.
	\newblock Congestion and centrality in traffic flow on complex networks.
	\newblock {\em Advances in Complex Systems}, 6(02):163--176, 2003.
	
	\bibitem{huisinga2001microscopic}
	Torsten Huisinga, Robert Barlovic, Wolfgang Knospe, Andreas Schadschneider, and
	Michael Schreckenberg.
	\newblock A microscopic model for packet transport in the internet.
	\newblock {\em Physica A: Statistical Mechanics and its Applications},
	294(1-2):249--256, 2001.
	
	\bibitem{huitema2000routing}
	Christian Huitema.
	\newblock {\em Routing in the Internet}.
	\newblock Prentice-Hall,, 2000.
	
	\bibitem{interdonato2016trust}
	Roberto Interdonato and Andrea Tagarelli.
	\newblock To trust or not to trust lurkers?: Evaluation of lurking and
	trustworthiness in ranking problems.
	\newblock In {\em International Conference and School on Network Science},
	pages 43--56. Springer, 2016.
	
	\bibitem{jang2012optimal}
	Young~Jae Jang, Young~Dae Ko, and Seungmin Jeong.
	\newblock Optimal design of the wireless charging electric vehicle.
	\newblock In {\em Electric Vehicle Conference (IEVC), 2012 IEEE International},
	pages 1--5. IEEE, 2012.
	
	\bibitem{jayasinghe2015explaining}
	Amila Jayasinghe, Kazushi Sano, and Hiroaki Nishiuchi.
	\newblock Explaining traffic flow patterns using centrality measures.
	\newblock {\em International Journal for Traffic and Transport Engineering},
	5(2):134--149, 2015.
	
	\bibitem{jayaweeracentrality}
	IMLN Jayaweera, KKKR Perera, and J~Munasinghe.
	\newblock Centrality measures to identify traffic congestion on road networks:
	A case study of sri lanka.
	\newblock {\em IOSR Journal of Mathematics (IOSR-JM)}, 2017.
	
	\bibitem{jiang2004structural}
	Bin Jiang and Christophe Claramunt.
	\newblock A structural approach to the model generalization of an urban street
	network.
	\newblock {\em GeoInformatica}, 8(2):157--171, 2004.
	
	\bibitem{jiang2011agent}
	Bin Jiang and Tao Jia.
	\newblock Agent-based simulation of human movement shaped by the underlying
	street structure.
	\newblock {\em International Journal of Geographical Information Science},
	25(1):51--64, 2011.
	
	\bibitem{katz1998luring}
	Jon Katz.
	\newblock Luring the lurkers.
	\newblock {\em Retrieved March}, 1(1999):1999, 1998.
	
	\bibitem{katz1953new}
	Leo Katz.
	\newblock A new status index derived from sociometric analysis.
	\newblock {\em Psychometrika}, 18(1):39--43, 1953.
	
	\bibitem{keeling2011modeling}
	Matt~J Keeling and Pejman Rohani.
	\newblock {\em Modeling infectious diseases in humans and animals}.
	\newblock Princeton University Press, 2011.
	
	\bibitem{kendall1938new}
	Maurice~G Kendall.
	\newblock A new measure of rank correlation.
	\newblock {\em Biometrika}, 30(1/2):81--93, 1938.
	
	\bibitem{kephart1997fighting}
	Jeffrey~O Kephart, Gregory~B Sorkin, David~M Chess, and Steve~R White.
	\newblock Fighting computer viruses.
	\newblock {\em Scientific American}, 277(5):88--93, 1997.
	
	\bibitem{khan2017utility}
	MD~Khan, Mashrur Chowdhury, Sakib~Mahmud Khan, Ilya Safro, and Hayato
	Ushijima-Mwesigwa.
	\newblock Utility maximization framework for opportunistic wireless electric
	vehicle charging.
	\newblock {\em Transportation Research Board 97th Annual Meeting,
		Transportation Research Board}, 2018.
	
	\bibitem{khan2019wireless}
	Zadid Khan, Sakib~Mahmud Khan, Mashrur Chowdhury, Ilya Safro, and Hayato
	Ushijima-Mwesigwa.
	\newblock Wireless charging utility maximization and intersection control delay
	minimization framework for electric vehicles.
	\newblock {\em {accepted in Computer-Aided Civil and Infrastructure
			Engineering}}, 2019.
	
	\bibitem{kitsak2010identification}
	Maksim Kitsak, Lazaros~K Gallos, Shlomo Havlin, Fredrik Liljeros, Lev Muchnik,
	H~Eugene Stanley, and Hern{\'a}n~A Makse.
	\newblock Identification of influential spreaders in complex networks.
	\newblock {\em Nature physics}, 6(11):888, 2010.
	
	\bibitem{kleinfeld2002small}
	Judith~S Kleinfeld.
	\newblock The small world problem.
	\newblock {\em Society}, 39(2):61--66, 2002.
	
	\bibitem{klemm2012measure}
	Konstantin Klemm, M~{\'A}ngeles Serrano, V{\'\i}ctor~M Egu{\'\i}luz, and Maxi
	San~Miguel.
	\newblock A measure of individual role in collective dynamics.
	\newblock {\em Scientific reports}, 2:292, 2012.
	
	\bibitem{lai2014knowledge}
	Hui-Min Lai and Tsung~Teng Chen.
	\newblock Knowledge sharing in interest online communities: A comparison of
	posters and lurkers.
	\newblock {\em Computers in Human Behavior}, 35:295--306, 2014.
	
	\bibitem{leskovec2007dynamics}
	Jure Leskovec, Lada~A Adamic, and Bernardo~A Huberman.
	\newblock The dynamics of viral marketing.
	\newblock {\em ACM Transactions on the Web (TWEB)}, 1(1):5, 2007.
	
	\bibitem{leskovec2007graph}
	Jure Leskovec, Jon Kleinberg, and Christos Faloutsos.
	\newblock Graph evolution: Densification and shrinking diameters.
	\newblock {\em ACM Transactions on Knowledge Discovery from Data (TKDD)},
	1(1):2, 2007.
	
	\bibitem{li2015percolation}
	Daqing Li, Bowen Fu, Yunpeng Wang, Guangquan Lu, Yehiel Berezin, H~Eugene
	Stanley, and Shlomo Havlin.
	\newblock Percolation transition in dynamical traffic network with evolving
	critical bottlenecks.
	\newblock {\em Proceedings of the National Academy of Sciences},
	112(3):669--672, 2015.
	
	\bibitem{li2015wireless}
	Siqi Li and Chunting~Chris Mi.
	\newblock Wireless power transfer for electric vehicle applications.
	\newblock {\em IEEE journal of emerging and selected topics in power
		electronics}, 3(1):4--17, 2015.
	
	\bibitem{liu2002simple}
	F~Liu, Y~Ren, and XM~Shan.
	\newblock A simple cellular automata model for packet transport in the
	internet.
	\newblock {\em Acta Physica Sinica}, 51(6):1175--1180, 2002.
	
	\bibitem{liu2016locating}
	Jian-Guo Liu, Jian-Hong Lin, Qiang Guo, and Tao Zhou.
	\newblock Locating influential nodes via dynamics-sensitive centrality.
	\newblock {\em Scientific reports}, 6:21380, 2016.
	
	\bibitem{liu2007opinion}
	Jian-Guo Liu, Zhi-Xi Wu, and Feng Wang.
	\newblock Opinion spreading and consensus formation on square lattice.
	\newblock {\em International Journal of Modern Physics C}, 18(07):1087--1094,
	2007.
	
	\bibitem{lukic2013cutting}
	Srdjan Lukic and Zeljko Pantic.
	\newblock Cutting the cord: Static and dynamic inductive wireless charging of
	electric vehicles.
	\newblock {\em {IEEE Electrification Magazine}}, 1(1):57--64, 2013.
	
	\bibitem{marett2009decision}
	Kent Marett and Kshiti~D Joshi.
	\newblock The decision to share information and rumors: Examining the role of
	motivation in an online discussion forum.
	\newblock {\em Communications of the Association for Information Systems},
	24(1):4, 2009.
	
	\bibitem{mason1999issues}
	Bruce Mason.
	\newblock Issues in virtual ethnography.
	\newblock {\em Ethnographic studies in real and virtual environments: Inhabited
		information spaces and connected communities}, pages 61--69, 1999.
	
	\bibitem{moreno2004dynamics}
	Yamir Moreno, Maziar Nekovee, and Amalio~F Pacheco.
	\newblock Dynamics of rumor spreading in complex networks.
	\newblock {\em {Physical Review E}}, 69(6):066130, 2004.
	
	\bibitem{motter2004cascade}
	Adilson~E Motter.
	\newblock Cascade control and defense in complex networks.
	\newblock {\em Physical Review Letters}, 93(9):098701, 2004.
	
	\bibitem{nagel1996particle}
	Kai Nagel.
	\newblock Particle hopping models and traffic flow theory.
	\newblock {\em Physical review E}, 53(5):4655, 1996.
	
	\bibitem{newman2003structure}
	Mark~EJ Newman.
	\newblock The structure and function of complex networks.
	\newblock {\em SIAM review}, 45(2):167--256, 2003.
	
	\bibitem{newman2005measure}
	Mark~EJ Newman.
	\newblock A measure of betweenness centrality based on random walks.
	\newblock {\em Social networks}, 27(1):39--54, 2005.
	
	\bibitem{ning2013compact}
	Puqi Ning, John~M Miller, Omer~C Onar, and Clifford~P White.
	\newblock A compact wireless charging system for electric vehicles.
	\newblock In {\em {Energy Conversion Congress and Exposition (ECCE), 2013
			IEEE}}, pages 3629--3634. IEEE, 2013.
	
	\bibitem{page1999pagerank}
	Lawrence Page, Sergey Brin, Rajeev Motwani, and Terry Winograd.
	\newblock The pagerank citation ranking: Bringing order to the web.
	\newblock Technical report, Stanford InfoLab, 1999.
	
	\bibitem{park2010social}
	Kyoungjin Park and Alper Yilmaz.
	\newblock A social network analysis approach to analyze road networks.
	\newblock In {\em {ASPRS Annual Conference. San Diego, CA}}, 2010.
	
	\bibitem{pastor2001epidemic}
	Romualdo Pastor-Satorras and Alessandro Vespignani.
	\newblock Epidemic spreading in scale-free networks.
	\newblock {\em Physical review letters}, 86(14):3200, 2001.
	
	\bibitem{porta2006network}
	Sergio Porta, Paolo Crucitti, and Vito Latora.
	\newblock The network analysis of urban streets: a primal approach.
	\newblock {\em Environment and Planning B: planning and design},
	33(5):705--725, 2006.
	
	\bibitem{preece2004top}
	Jenny Preece, Blair Nonnecke, and Dorine Andrews.
	\newblock The top five reasons for lurking: improving community experiences for
	everyone.
	\newblock {\em Computers in human behavior}, 20(2):201--223, 2004.
	
	\bibitem{qiu2013overview}
	Chun Qiu, KT~Chau, Chunhua Liu, and CC~Chan.
	\newblock Overview of wireless power transfer for electric vehicle charging.
	\newblock In {\em {Electric Vehicle Symposium and Exhibition (EVS27), 2013
			World}}, pages 1--9. IEEE, 2013.
	
	\bibitem{riemann2015optimal}
	Raffaela Riemann, David~ZW Wang, and Fritz Busch.
	\newblock Optimal location of wireless charging facilities for electric
	vehicles: flow-capturing location model with stochastic user equilibrium.
	\newblock {\em Transportation Research Part C: Emerging Technologies},
	58:1--12, 2015.
	
	\bibitem{ripeanu2002mapping}
	Matei Ripeanu and Ian Foster.
	\newblock {Mapping the Gnutella network: Macroscopic properties of large-scale
		peer-to-peer systems}.
	\newblock In {\em {International Workshop on Peer-to-Peer systems}}, pages
	85--93. Springer, 2002.
	
	\bibitem{rossi2015network}
	Ryan Rossi and Nesreen Ahmed.
	\newblock The network data repository with interactive graph analytics and
	visualization.
	\newblock In {\em {AAAI}}, volume~15, pages 4292--4293, 2015.
	
	\bibitem{scheurer2008spatial}
	Jan Scheurer, Carey Curtis, and S~Porta.
	\newblock {\em Spatial Network Analysis of Multimodal Transport Systems:
		Developing a Strategic Planning Tool to Assess the Congruence of Movement and
		Urban Structure: a Case Study of Perth Before and After the Perth-to-Mandurah
		Railway}.
	\newblock GAMUT, Australasian Centre for the Governance and Management of Urban
	Transport, University of Melbourne, 2008.
	
	\bibitem{schlosser2005posting}
	Ann~E Schlosser.
	\newblock Posting versus lurking: Communicating in a multiple audience context.
	\newblock {\em Journal of Consumer Research}, 32(2):260--265, 2005.
	
	\bibitem{shaydulin2017relaxation}
	Ruslan Shaydulin, Jie Chen, and Ilya Safro.
	\newblock Relaxation-based coarsening for multilevel hypergraph partitioning.
	\newblock {\em accepted in {SIAM Multiscale Modeling and Simulation}, preprint
		arXiv:1710.06552}, 2017.
	
	\bibitem{vsikic2013epidemic}
	Mile {\v{S}}iki{\'c}, Alen Lan{\v{c}}i{\'c}, Nino Antulov-Fantulin, and Hrvoje
	{\v{S}}tefan{\v{c}}i{\'c}.
	\newblock Epidemic centrality—is there an underestimated epidemic impact of
	network peripheral nodes?
	\newblock {\em {The European Physical Journal B}}, 86(10):440, 2013.
	
	\bibitem{soroka2003we}
	Vladimir Soroka, Michal Jacovi, and Sigalit Ur.
	\newblock We can see you: a study of communities’ invisible people through
	reachout.
	\newblock In {\em Communities and technologies}, pages 65--79. Springer, 2003.
	
	\bibitem{southworth2013streets}
	Michael Southworth and Eran Ben-Joseph.
	\newblock {\em Streets and the Shaping of Towns and Cities}.
	\newblock Island Press, 2013.
	
	\bibitem{rocketfuel}
	N.~Spring, R.~Mahajan, and D.~Wetherall.
	\newblock Measuring {ISP} topologies with rocketfuel.
	\newblock In {\em {SIGCOMM}}, volume~32, pages 133--145, 2002.
	
	\bibitem{staudt2017generating}
	Christian~L Staudt, Michael Hamann, Alexander Gutfraind, Ilya Safro, and
	Henning Meyerhenke.
	\newblock Generating realistic scaled complex networks.
	\newblock {\em Applied Network Science}, 2(1):36, 2017.
	
	\bibitem{strogatz2001exploring}
	Steven~H Strogatz.
	\newblock Exploring complex networks.
	\newblock {\em nature}, 410(6825):268, 2001.
	
	\bibitem{tagarelli2013s}
	Andrea Tagarelli and Roberto Interdonato.
	\newblock Who's out there?: identifying and ranking lurkers in social networks.
	\newblock In {\em {Proceedings of the 2013 IEEE/ACM International Conference on
			Advances in Social Networks Analysis and Mining}}, pages 215--222. ACM, 2013.
	
	\bibitem{tagarelli2014lurking}
	Andrea Tagarelli and Roberto Interdonato.
	\newblock Lurking in social networks: topology-based analysis and ranking
	methods.
	\newblock {\em Social Network Analysis and Mining}, 4(1):230, 2014.
	
	\bibitem{tao2006epidemic}
	Zhou Tao, Fu~Zhongqian, and Wang Binghong.
	\newblock Epidemic dynamics on complex networks.
	\newblock {\em Progress in Natural Science}, 16(5):452--457, 2006.
	
	\bibitem{tsoumakos2006analysis}
	Dimitrios Tsoumakos and Nick Roussopoulos.
	\newblock Analysis and comparison of p2p search methods.
	\newblock In {\em Proceedings of the 1st international conference on Scalable
		information systems}, page~25. ACM, 2006.
	
	\bibitem{ushijima2017optimal}
	Hayato Ushijima-Mwesigwa, MD~Khan, Mashrur~A Chowdhury, and Ilya Safro.
	\newblock Optimal installation for electric vehicle wireless charging lanes.
	\newblock {\em arXiv preprint arXiv:1704.01022}, 2017.
	
	\bibitem{van20141}
	Trevor Van~Mierlo.
	\newblock The 1\% rule in four digital health social networks: an observational
	study.
	\newblock {\em Journal of medical Internet research}, 16(2), 2014.
	
	\bibitem{vilathgamuwa2015wireless}
	DM~Vilathgamuwa and JPK Sampath.
	\newblock Wireless power transfer for electric vehicles, present and future
	trends.
	\newblock In {\em Plug in electric vehicles in smart grids}, pages 33--60.
	Springer, 2015.
	
	\bibitem{wang2012understanding}
	Pu~Wang, Timothy Hunter, Alexandre~M Bayen, Katja Schechtner, and Marta~C
	Gonz{\'a}lez.
	\newblock Understanding road usage patterns in urban areas.
	\newblock {\em Scientific reports}, 2:1001, 2012.
	
	\bibitem{wang2007analyzing}
	Yong Wang, Xiaochun Yun, and Yifei Li.
	\newblock Analyzing the characteristics of gnutella overlays.
	\newblock In {\em {Information Technology, 2007. ITNG'07. Fourth International
			Conference on}}, pages 1095--1100. IEEE, 2007.
	
	\bibitem{watts2007viral}
	Duncan~J Watts, Jonah Peretti, and Michael Frumin.
	\newblock {\em Viral marketing for the real world}.
	\newblock Harvard Business School Pub., 2007.
	
	\bibitem{wood1993deterministic}
	R~Kevin Wood.
	\newblock Deterministic network interdiction.
	\newblock {\em Mathematical and Computer Modelling}, 17(2):1--18, 1993.
	
	\bibitem{yan2006efficient}
	Gang Yan, Tao Zhou, Bo~Hu, Zhong-Qian Fu, and Bing-Hong Wang.
	\newblock Efficient routing on complex networks.
	\newblock {\em Physical Review E}, 73(4):046108, 2006.
	
	\bibitem{zhang2011centrality}
	Yuanyuan Zhang, Xuesong Wang, Peng Zeng, and Xiaohong Chen.
	\newblock Centrality characteristics of road network patterns of traffic
	analysis zones.
	\newblock {\em Transportation Research Record: Journal of the Transportation
		Research Board}, (2256):16--24, 2011.
	
	\bibitem{zhao2013rumor}
	Laijun Zhao, Wanlin Xie, H~Oliver Gao, Xiaoyan Qiu, Xiaoli Wang, and Shuhai
	Zhang.
	\newblock A rumor spreading model with variable forgetting rate.
	\newblock {\em {Physica A: Statistical Mechanics and its Applications}},
	392(23):6146--6154, 2013.
	
\end{thebibliography}
\end{document}